\documentclass[aps,pra,10pt,twocolumn]{revtex4-2}

\usepackage{mathrsfs}
\usepackage{graphicx}
\usepackage{comment}
\usepackage{dsfont}
\usepackage{amsmath}
\usepackage{amsfonts}
\usepackage{amssymb}
\usepackage{braket}
\usepackage{amsthm}
\usepackage[colorlinks=true,urlcolor=blue,citecolor=blue,linkcolor=blue]{hyperref}
\usepackage{caption}
\usepackage{lipsum}
\usepackage{subcaption}
\usepackage{mwe}
\usepackage{enumerate}
\usepackage{sidecap}
\usepackage{natbib}
\usepackage{appendix}

\makeatletter
\def\maketitle{
\@author@finish
\title@column\titleblock@produce
\suppressfloats[t]}
\makeatother

\usepackage{graphicx}
\usepackage{dcolumn}
\usepackage{amsmath}
\usepackage{bm}
\usepackage[colorlinks=true,urlcolor=blue,citecolor=blue,linkcolor=blue]{hyperref}
\usepackage{array}

\newcommand{\fo}{\rightarrowtail}
\newcommand{\setFO}{\mathcal{FO}}
\newcommand{\setCO}{\mathcal{CO}}
\newcommand{\setCCO}{\mathcal{CCO}}
\newcommand{\mono}{\mathfrak{R}}

\newcommand\T{\rule{0pt}{2.4ex}}       
\newcommand\B{\rule[-1.ex]{0pt}{0pt}}

\newcommand{\cS}{\mathcal{S}}
\newcommand{\cH}{\mathcal{H}}
\newcommand{\states}{\cS(\cH)}
\newcommand{\substates}{\cS_\bullet(\cH)}

\newtheorem{theorem}{Theorem}
\newtheorem{lemma}[theorem]{Lemma}
\newtheorem{proposition}[theorem]{Proposition}

\newtheorem{definition}[theorem]{Definition}

\newtheorem{corollary}[theorem]{Corollary}

\usepackage{multirow}

\begin{document}

\title{Fundamental Limits on Correlated Catalytic State Transformations}

\author{Roberto Rubboli}
\email{roberto.rubboli@u.nus.edu}
\affiliation{Centre for Quantum Technologies, National University of Singapore, Singapore}

\author{Marco Tomamichel}
\affiliation{Department of Electrical and Computer Engineering, National University of Singapore, Singapore 117583, Singapore}
\affiliation{Centre for Quantum Technologies, National University of Singapore, Singapore}

\begin{abstract}
Determining whether a given state can be transformed into a target state using free operations is one of the fundamental questions in the study of resources theories. Free operations in resource theories can be enhanced by allowing for a catalyst system that assists the transformation and is returned unchanged, but potentially correlated, with the target state. While this has been an active area of recent research, very little is known about the necessary properties of such catalysts. Here, we prove fundamental limits applicable to a large class of correlated catalytic transformations by showing that a small residual correlation between catalyst and target state implies that the catalyst needs to be highly resourceful. In fact, the resources required diverge in the limit of vanishing residual correlation. In addition, we establish that in imperfect catalysis a small error generally implies a highly resourceful embezzling catalyst. We develop our results in a general resource theory framework and discuss its implications for the resource theory of athermality, the resource theory of coherence and entanglement theory.
\end{abstract}

\maketitle


\paragraph*{Introduction.} A quantum resource theory is defined by a set of \emph{free operations}~\footnote{For our purposes, quantum operations or quantum channels are completely positive and trace-preserving maps from linear operators on an input to linear operators on an output Hilbert space.} and a set of \emph{free states} with the property that free operations are closed under composition and map free states into free states~\cite{Gour,Brandao2}. Resource theories offer a general and versatile framework to quantify the usefulness of different quantum states and their interconvertibility using free operations.
Prominent examples of resource theories include entanglement theory~\cite{horodecki2009quantum,Plenio,Vlatko} (where local operations and classical communication are free and entanglement is considered a resource), athermality in thermodynamics~\cite{Brandao3,faist2015gibbs,Oppenheim} (where transformations that preserve the thermal state are free and states out of thermal equilibrium are resourceful), and coherence~\cite{Winter,aberg2006quantifying,baumgratz2014quantifying} (where incoherent states are free and coherence is a resource).

In the following, we will use the notation $\rho \fo \rho'$ to indicate that a free transformation exists which maps a quantum state $\rho$ to a quantum state $\rho'$. Given a fixed state $\rho$, a fundamental question in any resource theory is to find the set of states $\setFO(\rho)$ of all $\rho'$ such that $\rho \fo \rho'$, i.e., all states that can be reached from $\rho$ using free operations. More precisely, we are often interested in its closure, $\overline{\setFO(\rho)}$, which also contains quantum states that can be arbitrarily well approximated by free operations from $\rho$~\footnote{We consider the closure within the same fixed output space. Note that other definitions might lead to different necessary and sufficient conditions~\cite{aubrun2008catalytic,Klimesh}.}.
The set of free operations can be enlarged by allowing for \emph{catalytic transformations}, $\rho \otimes \nu \fo \rho' \otimes \nu$, where the catalyst $\nu$ is returned unchanged. The set $\setCO(\rho)$ then contains all states $\rho'$ for which such a catalytic transformation from $\rho$ exists. Its closure is denoted $\overline{\setCO(\rho)}$.
More recently, a further relaxation has been studied where correlations between the catalyst and the target state after the transformation are allowed and can be used as a resource in catalytic transformations~\cite{lostaglio2015stochastic,muller2016generalization,sapienza2019correlations,muller2018correlating,Wilde,Wilming,Sagawa,Kondra,Takagi}.  
We say that a state $\rho'$ can be reached by a \emph{correlated catalytic transformation}, or $\rho' \in \setCCO(\rho)$, if there exists a catalyst $\nu$ such that $\rho \otimes \nu \fo \tau$ where $\tau$ is any state that has marginals $\rho'$ (for the target system) and $\nu$ (for the catalyst system). We will similarly be concerned with its closure, $\overline{\setCCO(\rho)}$. 

To the best of our knowledge, the idea of residual correlations between the system and the catalyst in the output state while the catalyst returns exactly to its original form was first introduced in~\cite{gallego2016thermodynamic}. In~\cite{wilming2017axiomatic} the authors first discussed whether the free energy completely characterizes correlated catalytic transformations in resource theory of athermality. This question was answered positively in the classical case and conjectured for the quantum case in~\cite{muller2018correlating}.  The conjecture for the quantum case has been recently resolved in the affirmative in~\cite{Sagawa} using the previously known construction that allows to reduce the problem to asymptotic interconvertibiliy~\cite{duan2005multiple}. This was recently generalized for any resource theory in~\cite{Takagi}.

The sets $\setFO, \setCO$ and $\setCCO$ are generally difficult to characterise, but they take on a natural form for certain resource theories where they are fully characterised by \emph{resource monotones}. Let $\mono{}$ be a function from quantum states to positive reals that measures the resourcefullness of states. We say that such a map is a) a resource monotone if it is non-increasing under free operations, b) \emph{tensor-additive} if it is additive under tensor-products, and c) \emph{super-additive} if $\mono(\rho_{AB}) \geq \mono(\rho_A) + \mono(\rho_B)$ for any joint state $\rho_{AB}$ with marginals $\rho_A$ and $\rho_B$. Resource monotones play an important role in characterizing the above sets.
It is easy to see that a necessary (but not generally sufficient) conditions for $\rho'$ to be in the set $\setFO(\rho)$ is that $\mono(\rho) \geq \mono(\rho')$ for any resource monotone. For $\rho'$ to be in $\setCO(\rho)$ this ordering only needs to be required for tensor-additive resource monotones, and finally for $\rho'$ to be in $\setCCO(\rho)$ the ordering only needs to be satisfied for tensor-additive and super-additive resource monotones. Finally, for $\rho'$ to be in the closure of the sets we require in addition that the resource monotone is lower semicontinuous (see Supplemental Material~\cite[Section IV, Lemma 6]{Supplemental_Material} for a proof). In general, it is not known which resource monotones characterise these sets, i.e. what are the necessary and sufficient conditions for $\rho' $ to be in any of the sets. The particular appeal of $\overline{\setCCO(\rho)}$ is that for some prominent resource theories it is fully characterised by a single resource monotone, e.g., the non-equilibrium free energy~\cite{Sagawa} or the relative entropy of entanglement~\cite{Kondra}. Moreover, the set $\overline{\setCCO(\rho)}$ is also of operational interest since it contains states that are strictly more useful than $ \overline{\setCO(\rho)}$ for some information-theoretic tasks, for example quantum teleportation~\cite{Lipka}.
FIG.~\ref{fig: CC 1} gives an example of these sets and their full characterisation for the resource theory of athermality restricted to states that commute with the Hamiltonian.

\tabcolsep=4.2pt
\begin{figure}[t]
\hskip-25pt \includegraphics[width=.54\textwidth]{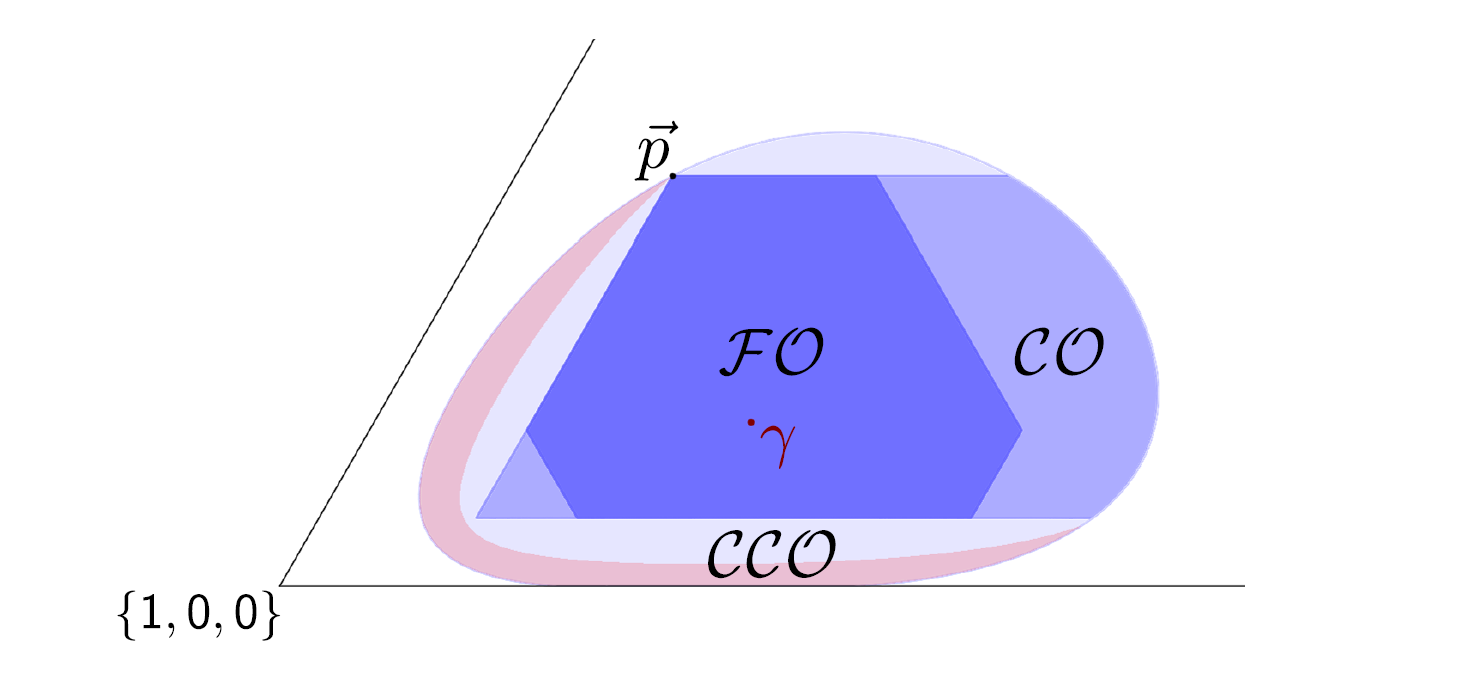}
\begin{tabular}{ |c| l |}
\hline
\multicolumn{2}{|c|}{Classical resource theory of athermality} \T\B \\
\hline
\hline
$ \mathcal{FO} $  & \multirow{2}{*}{$(p,\gamma) \succ (p',\gamma) $~\cite{Oppenheim}} \T \\
$ \overline{\mathcal{FO}} $ &  \B \\
\hline
\multirow{2}{*}{$\mathcal{CO}$}  & $ D_{\alpha}(p \| \gamma) > D_{\alpha}(p' \| \gamma) \,\, \text{and} \,\, D_{\alpha}(\gamma \| p) > D_{\alpha}(\gamma \| p')$ \T  \\ &  $\text{for all}\,\, \alpha \geq 1/2 \text{ and $p'$ has full support}$~\cite{Klimesh} \B   \\
\hline
\multirow{2}{*}{$\overline{\mathcal{CO}}$} & $ D_{\alpha}(p \| \gamma) \geq D_{\alpha}(p' \| \gamma) \,\, \text{and} \,\, D_{\alpha}(\gamma \| p) \geq D_{\alpha}(\gamma \| p')$  \T \\ &  $\text{for all}\,\, \alpha \geq 1/2$~\cite{Klimesh,Brandao} \B  \\
\hline
$\mathcal{CCO}$ & $D(p\|\gamma) > D(p'\| \gamma) \text{ and $p'$ has full support}$~\cite{muller2018correlating,Wilde} \T\B  \\ 
\hline
 $ \overline{\mathcal{CCO}}$ & $D(p \|\gamma) \geq D(p' \| \gamma)$~\cite{muller2018correlating} \T\B  \\  
\hline
\end{tabular}
\caption{Example of the sets $\setFO$, $\setCO$ and $\setCCO$ for classical resource theory of athermality with rational Gibbs states where we fixed the input state $\vec{p}=\{2/3, 1/12, 3/12\}$ and $\gamma = \{7/10, 2/10, 1/10\}$. We show one corner of the probability simplex (which is a triangle in this case). Each point in the triangle corresponds to a (classical) state of a three-dimensional system.  The points in the red region satisfy the conditions of Theorem~\ref{fidelity theorem}. The table contains the conditions characterizing each set, where $D_\alpha$ is the R\'enyi divergence of order $\alpha$ and $D$ is the Kullbach-Leibler divergence.}
\label{fig: CC 1}
\end{figure}

While allowing arbitrary correlations between the catalyst and target state arguably goes against the spirit of catalysis, recent works~\cite{Sagawa,Kondra,Takagi,muller2018correlating} showed that for some prominent reversible resource theories target states in $\rho' \in \overline{\setCCO(\rho)}$ can be achieved with arbitrarily small correlations with the catalyst.

In this work we investigate the fundamental limits of such correlated catalytic transformations. Our results apply to any catalytic transformation between a given pair of `hard-to-transform' states and are applicable to any resource theory in which certain monotones are tensor additive. We focus on the problem of preparing suitable catalysts and we find that for some target states that lie in the set $\overline{\setCCO} \setminus \overline{\setCO}$, correlated catalytic transformations with small correlations require catalysts that are highly resourceful, and in fact, require unbounded resources in the limit of vanishing correlations. (See FIG.~\ref{fig: CC 1} for a depiction of such states.)
In particular, we show a quantitative trade-off between the error $\varepsilon$ achievable in the transformation and the resources needed for the catalyst.  




\paragraph*{Formal setting.} We denote by $\states$ the set of quantum states on a $d$-dimensional Hilbert space $\mathcal{H}$. 
We introduce the purified distance~\cite{tomamichel2010duality}, which for normalised states is defined as $P(\rho, \sigma) := \sqrt{1-F(\rho, \sigma)}$, where $F(\rho, \sigma) :=  ( \text{Tr}|\sqrt{\rho}\sqrt{\sigma}| )^2$ is the Uhlmann fidelity. The Umegaki relative entropy is defined as $D(\rho\||\sigma := \text{Tr}[\rho (\log{\rho} - \log{\sigma})]$. Since both the fidelity and the relative entropy satisfy a data-processing inequality under quantum channels, we can define resource monotones
\begin{align*}
	\mathfrak{D}(\rho) = \min_{\sigma \in \mathcal{F}} D(\rho\|\sigma) \quad \textrm{and} \quad \frak{D}_{1/2}(\rho) := -\log \frak{F}(\rho) 
\end{align*}
with $ \frak{F}(\rho) := \max_{\sigma \in \mathcal{F}} F(\rho\|\sigma)$. These are the limiting cases at $\alpha = 1$ and $\alpha = 1/2$, respectively, of a larger family of resource monotones,
$\mathfrak{D}_\alpha(\rho) := \min_{\sigma \in \mathcal{F}} \widetilde{D}_\alpha(\rho \| \sigma)$ where $\widetilde{D}_{\alpha}(\rho \| \sigma)$ is the sandwiched R\'enyi divergence~\cite{Muller, Wilde3,Tomamichel} and is defined for $\alpha \in [\frac{1}{2},1) \cup (1,\infty)$ as~\cite{Muller, Wilde3,Tomamichel}
\begin{align*}
\widetilde{D}_{\alpha}(\rho \| \sigma) := \frac{1}{\alpha-1}\log{\text{Tr}(\sigma^{\frac{1-\alpha}{2\alpha}}\rho \sigma^{\frac{1-\alpha}{2\alpha}})^{\alpha}}\,.
\end{align*}
We say that $\mathfrak{D}_\alpha$ is additive for the state $\rho$ if $\mathfrak{D}_\alpha(\rho \otimes \nu) = \mathfrak{D}_\alpha(\rho) + \mathfrak{D}_\alpha(\nu)$ for any catalyst state $\nu$.

We are now ready to define correlated catalytic transformations~\cite{Sagawa,Kondra,Takagi} as follows:

\begin{definition}
\label{Correlated catalytic transformations}
Let $\rho, \rho' \in \states$ be a pair of quantum states and $\varepsilon > 0$ a small positive constant. We say that $\rho$ can be transformed into $\rho'$ by an $\varepsilon$-correlated catalytic transformation if there exists a free operation $\mathcal{N}$ and a catalyst state $\nu \in \mathcal{S}(\mathcal{H'}) $ such that $\mathcal{N}(\rho \otimes \nu) = \tau$,   $\textup{Tr}_{\mathcal{H}}[\tau] =\nu$ and $P(\rho' \otimes \nu, \tau) \leq \varepsilon$.
If this holds for any $\varepsilon > 0$ we say that $\rho$ is transformable into $\rho'$ by a correlated catalytic transformation. 
\end{definition}

For the specific resource theories we consider, the quantity $\mathfrak{D}$ completely characterizes the set $\overline{\setCCO}$, namely the necessary and sufficient condition for $\rho'$ to be in $\overline{\setCCO(\rho)}$ is that $\mathfrak{D}(\rho) \geq \mathfrak{D}(\rho')$ (see the discussion of the individual resource theories below). Motivated by this, we identify $\mathfrak{D}$ as the relevant resource measure to evaluate the resourcefulness of the catalyst. We remark that the dimension of the Hilbert space of the catalyst, without adding any further constraints, does not quantify the resourcefulness of the catalyst. For example, in the resource theory of athermality, states with large free energy can be constructed easily in low dimension using a sufficiently gapped Hamiltonian.
\paragraph*{Main result and discussion.} 
We are now ready to state our main theorem.
\begin{theorem}
\label{fidelity theorem}
Assume that $\rho, \rho' \in \mathcal{S}(\mathcal{H})$ and $\alpha \in [1/2,1)$ such that $\mathfrak{D}_\alpha$ is additive for the state $\rho'$ and $\mathfrak{D}_\alpha(\rho) < \mathfrak{D}_\alpha(\rho')$. Then, for any $\varepsilon$-correlated catalytic transformation with catalyst $\nu$ mapping $\rho$ into $\rho'$, we have
\begin{align}
\mathfrak{D}(\nu) = \Omega \left(\log{\frac{1}{\varepsilon}} \right) \,. \nonumber
\end{align}
In particular, when $\alpha = 1/2$ and, thus, $\mathfrak{F}(\rho) > \mathfrak{F}(\rho')$, we have
the quantitative bound
\begin{align}
	\sqrt{\mathfrak{F}(\nu)} \leq \frac{\varepsilon}{\sqrt{\mathfrak{F}(\rho)}-\sqrt{\mathfrak{F}(\rho')}} \,. \nonumber
\end{align}
\end{theorem}

We formulated the above theorem for any resource theories but it is only meaningful when there exist a pair of states and $\alpha$ satisfying the assumptions. Quantitative bounds for $\alpha \neq 1/2$ can be found in the Supplemental Material~\cite[Section V]{Supplemental_Material}.
The quantities $\mathfrak{D}_\alpha$ can be interpreted as a measure of distance between a state and the free set. In the following, we refer to the quantity $\sqrt{\mathfrak{F}(\rho)}-\sqrt{\mathfrak{F}(\rho')}$ as fidelity gap.

The condition $\mathfrak{D}_\alpha(\rho) < \mathfrak{D}_\alpha(\rho') $ for some $\alpha \in [1/2,1)$, together with the additivity assumption, implies that the output state $\rho'$ lies outside the set $\overline{\setCO(\rho)}$ (see the Supplemental Material~\cite[Section IV, Lemma 6]{Supplemental_Material} for a detailed discussion). Hence, catalytic transformation from $\rho$ to $\rho'$ is possible only by allowing correlations (see FIG.~\ref{fig: CC 1}). For this reason, we say that the pair of states $(\rho, \rho')$ is `hard-to-transform'  and we will establish the existence of such state pairs for the resource theories we consider.

For correlated catalysis, i.e., when there are non-zero residual correlations between the catalyst and the system in the output state, the theorem implies that, as the error decreases, the distance between the catalyst and the free set must increase. In particular, in the limit of zero error, the catalyst state must be orthogonal to the set of free states, i.e. its resourcefulness is unbounded. As we discuss in the Supplemental Material~\cite[Section V]{Supplemental_Material} we can also derive bounds for the robustness of the catalyst.

We point out that the above Theorem actually holds also if we lift the restriction $\text{Tr}_{\mathcal{H}}[\tau] = \nu$ and hence we do not need to exactly recover the catalyst after the transformation.  If we allow a small error in the catalyst after the transformation, any state transformation is possible. This phenomenon is called \textit{embezzling}~\cite{VanDam, Brandao, Nelly}.  Our result shows that to achieve small errors we need a highly resourceful embezzling catalyst. In particular, we recover the optimal lower bound for embezzlement already established for entanglement theory~\cite{VanDam,leung2014characteristics,cleve2017perfect} and we extend it, in principle, to any resource theory.

\paragraph*{Sketch of the proof of Theorem~\ref{fidelity theorem}.} 
We only give a sketch of the proof below but leave the formal derivation to the Supplemental Material~\cite[Section V and Appendix A]{Supplemental_Material}.
We will need the \textit{smoothed sandwiched quantum R\'enyi divergence}, which is defined for two states $\rho,\sigma \in \states$ and $\alpha \in [1/2,1)$ as
\begin{align}
 \widetilde{D}^\varepsilon_{\alpha}(\rho \| \sigma) :=
      \max \left\{ \widetilde{D}_{\alpha}(\tilde{\rho}\|\sigma) :  \tilde{\rho} \in \substates ,\, P(\tilde{\rho}, \rho) \leq \varepsilon \right\} \nonumber ,
\end{align}
where $\substates$ is the set of sub-normalised states.
An important ingredient in the proof of the above theorem is the data-processing inequality for this quantity. 
We believe this result to be of independent interest. In the Supplemental Material~\cite[Section II and Appendix B]{Supplemental_Material} we give a proof and we also argue why a similar result does not hold for some other generalisations of R\'enyi divergence. We note that the use of sub-normalised states in the definition of the smoothed sandwiched quantum R\'enyi divergence turns out to be crucial for $\alpha \in [\frac12, 1)$, which is in contrast to the case $\alpha > 1$.
\begin{theorem}
\label{normal data-processing}
Let $\rho,\sigma \in\states$ be two states and $\mathcal{E}$ a quantum channel. For any $\alpha \in [1/2,1)$
$$\widetilde{D}^{\varepsilon}_{\alpha}(\rho\|\sigma) \geq \widetilde{D}^{\varepsilon}_{\alpha}(\mathcal{E}(\rho)\|\mathcal{E}(\sigma))$$ 
\end{theorem}
Another key ingredient of our proof is the following continuity bound for the quantum sandwiched R\'enyi divergences in the interval $\alpha \in (0,1)$.
\begin{proposition}
Let $\alpha \in (0,1)$ and $\rho, \sigma \in \mathcal{S}_\bullet(\mathcal{H})$. Then for any $\tilde{\rho} \in \mathcal{S}_\bullet(\mathcal{H})$ such that $\Delta(\rho, \tilde{\rho}) \leq \varepsilon \leq \widetilde{Q}_\alpha(\rho \| \sigma)^\frac{1}{\alpha}$ we have 
\begin{equation}
|\widetilde{D}_\alpha(\rho \| \sigma)-\widetilde{D}_\alpha(\tilde{\rho} \| \sigma)| \leq \frac{1}{\alpha-1} \log{\left(1- \frac{\varepsilon^\alpha}{\widetilde{Q}_\alpha(\rho \| \sigma) } \right)} 
\end{equation}
\end{proposition} 
where we introduced generalised trace distance~\cite{Tomamichel} which for normalised states is defined as $2\Delta(\rho, \sigma) := \|\rho-\sigma\|_1$ and the function $\widetilde{Q}_\alpha(\rho \| \sigma) = \exp{(\alpha-1)\widetilde{D}_\alpha(\rho \| \sigma)}$. We remark that the previous bound does not depend explicitly on the dimension of the Hilbert space of the states. Moreover, the previous proposition implies that the resource monotones $\mathfrak{D}_\alpha$ are also continuous~\cite[Corollary~4]{Supplemental_Material}.

\noindent The main idea of the proof of the main theorem is that we choose a pair of states $(\rho, \rho')$ such that $\mathfrak{D}_{\alpha}(\rho) < \mathfrak{D}_{\alpha}(\rho')$ 
and hence, since $\mathfrak{D}_{\alpha}$ is tensor-additive by assumption, the data-processing inequality for $\mathfrak{D}_{\alpha}$ for any (uncorrelated) catalytic transformation taking $\rho$ to $\rho'$ is strictly violated. Moreover, for any $\varepsilon$-correlated catalytic transformation with catalyst $\nu$ mapping $\rho$ into $\rho'$, we have both
\begin{align*}
\mathfrak{D}_{\alpha}(\rho \otimes \nu) &< \mathfrak{D}_{\alpha}(\rho' \otimes \nu) \, \quad \textnormal{and} \\
 \mathfrak{D}_{\alpha}^{\varepsilon}(\rho \otimes \nu) &\geq \mathfrak{D}_{\alpha}^{\varepsilon}(\tau) \geq \mathfrak{D}_{\alpha}(\rho' \otimes \nu) \nonumber ,
\end{align*}
where the inequalities on the second line are due to the monotonicity for the transformation (Theorem~\ref{normal data-processing}) including the catalyst and our assumption that $\tau$ is $\varepsilon$-close to $\rho' \otimes \nu$.

However, these two inequalities lead to a tension with the continuity of $\mathfrak{D}_{\alpha}$, which ensures that $\mathfrak{D}_{\alpha}(\rho \otimes \nu)$ and $\mathfrak{D}_{\alpha}^{\varepsilon}(\rho \otimes \nu)$ are arbitrarily close as $\varepsilon$ decreases. We then show that this tension can only be relieved if $\mathfrak{D}_{\alpha}(\nu)$ grows large when $\varepsilon$ decreases.

In the following, we quickly summarize the consequences of Theorem~\ref{fidelity theorem} for resource theory of athermality, entanglement theory, and resource theory of coherence. In each resource theory we will specify the set of free states; our results apply to any resource theory compatible with this choice of free states. To apply Theorem~\ref{fidelity theorem}, for each resource theory we first discuss the additivity of $\mathfrak{D}_\alpha$ and we then check if there exist states in $\overline{\setCCO}$ that satisfy the conditions of the theorem. In particular, in each resource theory we find states that both satisfy $\mathfrak{D}(\rho) \geq \mathfrak{D}(\rho')$ and $\mathfrak{F}(\rho)>\mathfrak{F}(\rho')$ .

\paragraph*{Resource theory of athermality.} 
In resource theory of athermality the thermal or Gibbs state $\gamma = e^{-\beta H} /Z$ is the only free state. Here, $\beta$ is the inverse temperature, $H$ is the Hamiltonian of the system and $Z$ is the normalisation factor (partition function). The relevant resource measure is the non-equilibrium free energy~\cite{Brandao3}, $\mathfrak{D}(\rho) = D(\rho\|\gamma)$.
We remark that our results apply to both resource theory of athermality with thermal operations and resource theory of athermality under Gibbs preserving maps since in both resource theories free operations keep the Gibbs state invariant~\cite{faist2015gibbs,lostaglio2018elementary}. However since for the former we do not know the resource monotone characterizing the set $\overline{\setCCO}$, for our considerations we will mainly focus on the latter where the relevant resource monotone is the non-equilibrium free energy~\cite{muller2018correlating,Sagawa}. 

The resource monotones $\mathfrak{D}_\alpha$ are trivially additive and we prove in the Supplemental Material~\cite[Section VI]{Supplemental_Material} that there exist states in $\overline{\setCCO}$ satisfying conditions of Theorem~\ref{fidelity theorem}. In particular, we find numerically pairs of qubit states with a non-zero fidelity gap and we construct analytically pairs of classical qutrit states with fidelity gap arbitrarily close to one.
 
From Theorem~\ref{fidelity theorem} we get that the non-equilibrium free energy of any catalyst must satisfy
$\mathfrak{D}(\nu) = \Omega\left( \log{\frac{1}{\varepsilon}}\right)
$.
Therefore, a correlated catalytic transformation between any two states would require preparing a catalyst with an unbounded amount of free energy as the error vanishes. Moreover, we show that the protocol discussed in~\cite{Sagawa} is optimal, i.e. $\mathfrak{D}(\nu) = \Theta\left( \log{\frac{1}{\varepsilon}}\right)
$. With this protocol, any state in $\overline{\setCCO}$ can be reached up to arbitrary accuracy. This means that it reaches exactly some states in $\setCCO \setminus \overline{\setCO}$ (with finite residual correlations for finite resourceful catalysts). This method generalizes to correlated catalytic transformations the already known construction of the catalyst introduced in~\cite{duan2005multiple} for (uncorrelated) catalytic transformations. This method provides a recipe to construct the catalyst whenever the states are asymptotically transformable.

To prove that it is optimal we use the exponential upper bound for the convergence of the error in approximate asymptotic pairwise state transformation~\cite{Buscemi}. In~\cite{Buscemi} the authors provided a method to obtain a lower bound for the error exponent which controls the exponential convergence of the error to zero with the number of copies. In this work, we give a qualitative first order expansion of the error exponent  for small relative entropy gaps. We find that, under some mild regularity conditions, the error exponent $\gamma$ satisfies $
\gamma \geq \Delta D^2 \log{e}/8(V_1+V_2)+O(\Delta D^3) 
$
where $\Delta D := D(\rho_1 \|\sigma_1) - D(\rho_2\|\sigma_2)$ is the relative entropy gap and $V_i := V(\rho_i\|\sigma_i) = \text{Tr}[\rho_i(\log{\rho_i}-\log{\sigma_i})^2)]-D(\rho_i\|\sigma_i)^2$ is the relative entropy variance. 
We remark that the above expression shows the appropriate scaling behavior with the number of copies of the states (see the Supplemental Material~\cite[Section VI and Appendix C]{Supplemental_Material} for more details).

\paragraph*{Entanglement theory.} In this case, the separable states are the free states of the theory~\cite{Vlatko,horodecki2009quantum}. In the following, we consider input and output bipartite pure states $|\psi _{AB}\rangle, |\psi' _{AB}\rangle$, but allow general mixed catalysts during the protocol. The resource monotone characterizing the set of pure states in $\overline{\setCCO}$ is the relative entropy of entanglement~\cite{Kondra}. Moreover, the resource monotones $\mathfrak{D}_\alpha$ are additive when one state is pure~\cite{Unpublished}. 

Therefore, our main theorem implies that for pairs of pure states satisfying the conditions of the main theorem any correlated catalytic transformation needs a catalyst with a diverging amount of relative entropy of entanglement as the error approaches zero. 
We then construct states with fidelity gap arbitrarily close to one.
(see the Supplemental Material~\cite[Section VII]{Supplemental_Material} for more details).

\paragraph*{Resource theory of coherence.} Fixing a basis $\{|i\rangle, i=1, ...,d\} $, we say that a state is free if it is diagonal in such a basis~\cite{Winter}.
We consider output pure states where the monotone that characterizes the set $\overline{\setCCO}$ is the relative entropy of coherence~\cite{Winter,Superadditivity,Takagi}. 
All the monotones $\mathfrak{D}_\alpha$ are additive~[\citenum{Hayashi}, Theorem 3].  In our work, we give an independent proof of additivity of $\mathfrak{D}_{1/2}$ by finding an  Alberti's form of the Fidelity of Coherence $\mathfrak{F}(\rho):=\max_{\sigma \in \mathcal{F}}F(\rho,\sigma)$ through semi-definite program (SDP) formulation

\begin{theorem}
Let $\rho \in \states$. The fidelity of coherence is the solution of the following minimisation problem 
\begin{align}
\mathfrak{F}(\rho) = \inf \limits_{R > 0} \textnormal{Tr}\, [\rho R^{-1}] \|\Delta(R)\|_{\infty} \nonumber
\end{align}
where $\Delta$  is the dephasing operator $\Delta(\cdot) = \sum_i |i\rangle \langle i | \cdot |i\rangle \langle i |$.
\end{theorem}
We believe that this result is of independent interest since it allows to efficiently compute this quantity for which, to the best of our knowledge, an analytic form is known only for pure states~\cite{Hayashi}. 

Also in this case, we prove the existence of states in $\overline{\setCCO}$ satisfying conditions of Theorem~\ref{fidelity theorem} with a fidelity gap arbitrarily close to one.
We obtain from Theorem~\ref{fidelity theorem} for the relative entropy of coherence of the catalyst
$
\mathfrak{D}(\nu) = D(\nu\|\Delta(\nu)) = \Omega\left( \log{\frac{1}{\varepsilon}}\right)
$.
Hence, we establish that to perform correlated catalytic transformation we would need, at least for some states, to prepare a catalyst with a diverging amount of coherence as the error~vanishes. (see the Supplemental Material \cite[Section VIII]{Supplemental_Material} for more details).

\paragraph*{Conclusion and open questions.} In this work we established that for some correlated catalytic processes a small residual correlation between the system and the catalyst implies a highly resourceful catalyst. We also show similarly how in the context of imperfect catalysis a small error is only possible with a highly resourceful embezzling catalyst. Our results apply to resource theories for which certain resource monotones are tensor-additive. We point out that a characterization of the sets $\overline{\setCO}$ and $\overline{\setCCO}$,  and therefore of the set $\overline{\setCCO} \setminus \overline{\setCO}$, is not known for many resource theories. Hence, the range of applicability of our main theorem and whether unbounded resources for the catalyst are required in such theories are still open questions.

\paragraph*{Acknowledgements.}
This research is supported by the National Research Foundation, Prime Minister's
Office, Singapore and the Ministry of Education, Singapore under the Research Centres of Excellence programme.
MT is also supported in part by NUS startup grants (R-263-000-E32-133 and R-263-000-E32-731).

\nocite{*}

\bibliography{my}

\clearpage
\newpage

\title{Supplemental material for \\ ``Fundamental Limits on Correlated Catalytic State Transformations"}

\author{Roberto Rubboli}
\email{roberto.rubboli@u.nus.edu}
\affiliation{Centre for Quantum Technologies, National University of Singapore, Singapore}

\author{Marco Tomamichel}
\affiliation{Department of Electrical and Computer Engineering, National University of Singapore, Singapore 117583, Singapore}
\affiliation{Centre for Quantum Technologies, National University of Singapore, Singapore}

\maketitle

\onecolumngrid

In this Supplemental Material we provide the proofs of the results presented in the main text.

\setcounter{theorem}{0}

\section{Notation}
We denote by $\mathcal{S}(\mathcal{H})$ the set of quantum states on a $d$-dimensional Hilbert space $\mathcal{H}$ and with $\mathcal{S}_\bullet(\mathcal{H})$ the set of all subnormalised states, i.e. positive operators with trace smaller than one. We introduce the purified distance~\cite{tomamichel2010duality} for sub-normalised states, $P(\rho, \sigma) := \sqrt{1-F(\rho, \sigma)}$ where $\sqrt{F(\rho, \sigma)} :=  \text{Tr}|\sqrt{\rho}\sqrt{\sigma}| + \sqrt{(1-\text{tr} \rho)(1-\text{tr}\sigma)}$ is a generalisation of the Uhlmann fidelity to sub-normalised states. Moreover, we define the generalised trace distance~\cite{Tomamichel}, $2\Delta(\rho, \sigma):=\text{Tr}|\rho-\sigma|+|\text{Tr}(\rho-\sigma)|$ which for normalised states reduces to the trace distance $2d(\rho, \sigma) :=\text{Tr}|\rho-\sigma|$.

Let $\alpha \in [\frac{1}{2},1) \cup (1,\infty)$ and positive operators $\rho$ and $\sigma$ with $\rho \neq 0$. Then the \textit{sandwiched quantum R\'enyi divergence} of $\sigma$ with $\rho$ is defined as~\cite{Muller, Wilde3,Tomamichel}
\begin{equation}
\widetilde{D}_{\alpha}(\rho \| \sigma) :=
\begin{cases}
\frac{1}{\alpha-1}\log{\text{Tr}\big(\sigma^{\frac{1-\alpha}{2\alpha}}\rho \sigma^{\frac{1-\alpha}{2\alpha}}\big)^{\alpha}} & \text{if}\; (\alpha<1 \wedge \rho \not \perp \sigma) \vee \rho \ll \sigma \\
 +\infty & \text{else}
\end{cases}
\end{equation}
The sandwiched quantum R\'{e}nyi divergence of order $1/2$ is therefore 
$
D_{1/2}(\rho \| \sigma) = -\log{F(\rho, \sigma)}
$
. In the limit $\alpha \rightarrow 1$ the sandwiched quantum R\'{e}nyi divergence converges to the  Umegaki relative entropy $D(\rho \| \sigma) = \text{Tr}[\rho (\log{\rho} - \log{\sigma})]$. In the limit $\alpha \rightarrow \infty$ the sandwiched quantum R\'{e}nyi divergence converges to the \textit{max-divergence}~\cite{Datta_rob2, Renner}
\begin{align}
D_{\max}(\rho \| \sigma) := \inf \{ \lambda \in \mathbb{R} : \rho \leq 2^{\lambda}\sigma \} 
.
\end{align}
We also define the function $\widetilde{Q}_\alpha(\rho \| \sigma) = \exp{(\alpha-1)\widetilde{D}_\alpha(\rho \| \sigma)}$.

We  call a function $\mathfrak{R} : \mathcal{S}(\mathcal{H}) \rightarrow [0, + \infty]$ a resource monotone if it does not increase under free operations, i.e., if $\mathfrak{R}(\rho) \geq \mathfrak{R}(\mathcal{E}(\rho))$ for any state $\rho$ and any free operation $\mathcal{E}$.
In addition, we say that $\mathfrak{R}$ is \textit{tensor-additive} if $\mathfrak{R}(\rho \otimes \sigma) = \mathfrak{R}(\rho) +  \mathfrak{R}(\sigma) $ and \textit{super-additive} if $\mathfrak{R}(\rho_{AB}) \geq \mathfrak{R}(\text{Tr}_A[\rho_{AB}]) +  \mathfrak{R}(\text{Tr}_B[\rho_{AB}]) $. \\ 
We define also the resource monotones
\begin{align}
&\mathfrak{D}_\alpha(\rho) := \min \limits_{\sigma \in \mathcal{F}} \widetilde{D}_\alpha(\rho \|\sigma) \quad \alpha \in \left[1/2, \infty \right), \\
&\mathfrak{D}(\rho) := \min \limits_{\sigma \in \mathcal{F}} D(\rho \|\sigma)\, ,\\
&\mathfrak{D}_{\max}(\rho) := \min \limits_{\sigma \in \mathcal{F}} D_{\max}(\rho \|\sigma) \, .
\end{align} 
We also define $\mathfrak{D}_\alpha(\rho) := \frac{1}{\alpha -1}\log{\mathcal{Q}_\alpha(\rho)}$ and we also often call $\mathcal{Q}_{1/2}:= \sqrt{\mathfrak{F}}$.
In the literature, the robustness and the generalised robustness are often introduced to quantify the resourcefulness of a state.
The monotone $\mathfrak{D}_{\max}$ is equal to the 'generalised log-robusteness' $\mathfrak{D}_{\max}(\rho) = \log{(1 +\mathfrak{R}_g(\rho))}:=L\mathfrak{R}_g(\rho)$~\cite{Steiner, Vidal, Datta_rob1, Datta_rob2} where the generalised robustness is given by
\begin{equation}
\mathfrak{R}_g(\rho) := \min \left\{ s\geq 0: \exists \omega \in \mathcal{S}(\mathcal{H}) \,\, \text{s.t} \,\, \frac{1}{1+s}\rho + \frac{s}{1+s}\omega \in \mathcal{F} \right\}\, .
\end{equation}
We first define the smoothed quantum sandwiched R\'enyi divergences as
\begin{equation}
\label{smoothed}
    \widetilde{D}^\varepsilon_{\alpha}(\rho \| \sigma) :=
    \begin{cases}
      \max \limits_{\tilde{\rho} \in B^{\varepsilon}(\rho)} \widetilde{D}_{\alpha}(\tilde{\rho}\|\sigma), & \text{if}\ \alpha \in [1/2,1) \\
       \min \limits_{\tilde{\rho} \in B^{\varepsilon}(\rho)} \widetilde{D}_{\alpha}(\tilde{\rho}\|\sigma), & \text{if}\ \alpha \in (1,\infty).
    \end{cases} 
  \end{equation}
  where  $B^\varepsilon(\rho) = \{\tilde{\rho} \in  \mathcal{S}_\bullet(\mathcal{H}) : P(\rho, \tilde{\rho}) \leq \varepsilon\}$ for $\varepsilon \in(0,1)$ is the set of all subnormalized states which are $\varepsilon$-close in purified distance to $\rho$.
The related resource monotones for $\alpha \in [\frac{1}{2},1) \cup (1,\infty)$ are
\begin{align}
\label{monotones}
 \mathfrak{D}^\varepsilon _{\alpha}(\rho):= \min \limits_{\sigma \in \mathcal{F}} \widetilde{D}^\varepsilon _{\alpha}(\rho \| \sigma) \, .
\end{align}

\section{Data-processing inequality for smoothed R\'enyi sandwiched divergences}
\label{section data-processing}
In this section we show that the smoothed sandwiched R\'enyi divergences satisfy data-processing inequality and that their related resource monotones are therefore non-increasing under free operations. 

The proof that the sandwiched quantum R\'enyi divergences in the range $(1,\infty]$ satisfy the data-processing inequality trivially follows form the data-processing inequality of the underlying R\'enyi divergence. In the following we prove that it holds also in the interval $\alpha \in [1/2,1)$. 

\begin{theorem}
Let be $\rho,\sigma$ two states and $\mathcal{E}$ a quantum channel. For any $\alpha \in [1/2,1)$
$$\widetilde{D}^{\varepsilon}_{\alpha}(\rho\|\sigma) \geq \widetilde{D}^{\varepsilon}_{\alpha}(\mathcal{E}(\rho)\|\mathcal{E}(\sigma))$$ 
\end{theorem}

\begin{proof}
To prove the result for general quantum channels we take advantage of the Stinespring dilation and hence it suffices to prove the result only for isometries and the partial trace. 

Let us first consider an isometry $U$. We define $\tilde{\tau}$ such that $\widetilde{D}^{\varepsilon}_{\alpha}(U\rho U^\dagger\|U\sigma U^\dagger) = \widetilde{D}_{\alpha}(\tilde{\tau} \|U\sigma U^\dagger)$. We note that we can always choose the maximiser to be a subnormalised state with support only in the image of $U$ which we denote $\text{Im}(U)$. Indeed, if we call $P = UU^\dagger$ the projector onto $\text{Im}(U)$ and by noting that $U\sigma U^\dagger \in \text{Im}(U)$ we have 
\begin{align}
\text{Tr}((U\sigma U^\dagger)^{\frac{1-\alpha}{2\alpha}}\tilde{\tau} (U\sigma U^\dagger)^{\frac{1-\alpha}{2\alpha}})^{\alpha} = \text{Tr}((U\sigma U^\dagger )^{\frac{1-\alpha}{2\alpha}}P \tilde{\tau} P ( U\sigma U^\dagger )^{\frac{1-\alpha}{2\alpha}})^{\alpha} 
\end{align}
We then denote $\hat{\tau} = P \tilde{\tau} P$ the projection of $\tilde{\tau}$ into $\text{Im}(U)$ and define $\hat{\rho} = U^\dagger \hat{\tau} U$. 
We get
\begin{equation}
\widetilde{D}^{\varepsilon}_{\alpha}(\rho\|\sigma) \geq \widetilde{D}_{\alpha}(\hat{\rho}\|\sigma) \geq \widetilde{D}_{\alpha}(U \hat{\rho} U^\dagger \|U \sigma U^\dagger) =  \widetilde{D}_{\alpha}(P\hat{\tau}P\|U \sigma U^\dagger) = \widetilde{D}_{\alpha}(\hat{\tau}\|U \sigma U^\dagger)= \widetilde{D}^\varepsilon_{\alpha}(U \rho U^\dagger \|U\sigma U^\dagger)
\end{equation}
The first inequality follows from data-processing of the purified distance under trace non-increasing completely positive maps for which $P(\hat{\rho}, \rho)\leq P(\hat{\tau}, U \rho U^\dagger)$ and hence $\hat{\rho}$ is in the $\varepsilon$-ball of $\rho$. The second inequality is a consequence of data-processing of the underlying sandwiched R\'enyi divergence.

For the partial trace we use~([\citenum{Tomamichel}, Corollary 3.14]) which states that 
given $\rho_{AB}$ with marginal $\rho_A$ and $\tilde{\rho}_A$ which satisfies $P(\rho_A, \tilde{\rho}_A)\leq \varepsilon$ we can always find $\tilde{\rho}_{AB}$ with marginal $\tilde{\rho}_A$ such that $P(\rho_{AB}, \tilde{\rho}_{AB}) \leq \varepsilon$. 
Therefore if we define $\tilde{\rho}_A$ the optimser $\widetilde{D}_{\alpha}(\tilde{\rho}_{A}\|\sigma_A) = \widetilde{D}^{\varepsilon}_{\alpha}(\rho_{A}\|\sigma_A)$ then 
\begin{equation}
\widetilde{D}^{\varepsilon}_{\alpha}(\rho_{A}\|\sigma_A) = \widetilde{D}_{\alpha}(\tilde{\rho}_{A}\|\sigma_A) \leq \widetilde{D}_{\alpha}(\tilde{\rho}_{AB}\|\sigma_{AB}) \leq \widetilde{D}^{\varepsilon}_{\alpha}(\rho_{AB}\|\sigma_{AB})
\end{equation}
where we choose $\tilde{\rho}_{AB}$ as discussed above.
\end{proof}

We remark that data-processing in particular implies invariance under embedding of the two states into a larger space. The optimisation over sub-normalised states is necessary for the smoothed sandwiched R\'enyi divergences with $\alpha \in [1/2,1)$ to be invariant under embedding in a larger space. (See the discussion in Appendix \ref{remark on smoothing}.) Moreover, we also remark that for $\alpha \in [0,1)$ it is not possible to define smoothed Petz R\'enyi divergences that satisfy the data-processing inequality in a similar fashion.

It is then straightforward to prove that the monotones \eqref{monotones} are non-increasing under free operations. 
Indeed we find
\begin{corollary}
\label{data-processing}
For any free operation $\mathcal{E}$ and any  $\alpha \in [\frac{1}{2},1) \cup (1,\infty)$ we have
\begin{align}
\mathfrak{D}^\varepsilon_{\alpha}(\rho) \geq \mathfrak{D}^\varepsilon_{\alpha}(\mathcal{E}(\rho))
\end{align}
\end{corollary}

\begin{proof}
We obtain from the definitions
$$\mathfrak{D}^\varepsilon_{\alpha}(\rho) = \min_{\sigma \in \mathcal{F}}D^\varepsilon_{\alpha}(\rho \| \sigma) \geq \min \limits_{\sigma\in \mathcal{F}} D^\varepsilon_{\alpha}(\mathcal{E}(\rho) \| \mathcal{E}(\sigma)) \geq \min \limits_{\sigma \in \mathcal{F}} D^\varepsilon_{\alpha}(\mathcal{E}(\rho) \| \sigma)  = \mathfrak{D}^\varepsilon_{\alpha}(\mathcal{E}(\rho))
$$
The first inequality follows from data-processing inequality of the smoothed sandwiched divergences and in the second inequality we used that since $\mathcal{E}$ is a free operations it holds $\mathcal{E}(\sigma) \in \mathcal{F}$.
\end{proof} 

\section{Continuity bound for sandwiched R\'enyi divergences}
\label{Continuiy bound}
In this section we derive a continuity bound for the sandwiched R\'enyi divergences in the interval $\alpha \in (0,1)$. 
\begin{proposition}
\label{continuity bound}
Let $\alpha \in (0,1)$ and $\rho, \sigma \in \mathcal{S}_\bullet(\mathcal{H})$. Then for any $\tilde{\rho} \in \mathcal{S}_\bullet(\mathcal{H})$ such that $\Delta(\rho, \tilde{\rho}) \leq \varepsilon \leq \tilde{Q}_\alpha(\rho \| \sigma)^\frac{1}{\alpha}$ we have 
\begin{equation}
\label{continuity bound_eq}
|\widetilde{D}_\alpha(\rho \| \sigma)-\widetilde{D}_\alpha(\tilde{\rho} \| \sigma)| \leq \frac{1}{\alpha-1} \log{\left(1- \frac{\varepsilon^\alpha}{\widetilde{Q}_\alpha(\rho \| \sigma) } \right)} \, .
\end{equation}
\end{proposition} 

\begin{proof}
Since the upper bound in the lemma is an increasing function of $\varepsilon$ we can assume the worst case scenario and set $\Delta(\rho, \tilde{\rho}) := \varepsilon$. We consider the most general case where both $\rho$ and $\tilde{\rho}$ are subnormalized states. We set $\rho - \tilde{\rho} = P'-Q'$ where $P'$ and $Q'$ are the positive and negative parts, respectively. We also define $\text{Tr}(\rho) = 1-\delta$ and $\text{Tr}(\tilde{\rho}) = 1-\tilde{\delta}$. We then have
\begin{align}
\label{C1}
&\tilde{\delta} - \delta = \text{Tr}(\rho - \tilde{\rho}) = \text{Tr}(P'-Q') \\
\label{C2}
& 2\varepsilon - |\tilde{\delta}-\delta| =  \text{Tr}(|\rho - \tilde{\rho}|) = \text{Tr}(P')+\text{Tr}(Q') = 2\text{Tr}(P') - (\tilde{\delta} - \delta)
\end{align}
where in the last equality of~\eqref{C1} we used~\eqref{C2}. 
It follows that $2\text{Tr}(P') = 2\varepsilon - |\tilde{\delta}-\delta| + (\tilde{\delta} - \delta) \leq 2 \varepsilon $.  We define the quantum state $P := P'/\text{Tr}(P')$. We then use that $ \rho \leq \rho + Q' = \tilde{\rho} + P' = \tilde{\rho} + \text{Tr}(P') P \leq \tilde{\rho} + \varepsilon P$ and we obtain
\begin{align}
&\rho \leq \tilde{\rho} + \varepsilon P \\
\label{Impl1}
\implies \quad & \text{Tr}[(\sigma^\frac{1-\alpha}{2\alpha} \rho \sigma^\frac{1-\alpha}{2\alpha})^\alpha] \leq \text{Tr}[(\sigma^\frac{1-\alpha}{2\alpha} (\tilde{\rho} + \varepsilon P) \sigma^\frac{1-\alpha}{2\alpha})^\alpha] \\
\label{Impl2}
\implies \quad &  \text{Tr}[(\sigma^\frac{1-\alpha}{2\alpha} \rho \sigma^\frac{1-\alpha}{2\alpha})^\alpha]
\leq \text{Tr}[(\sigma^\frac{1-\alpha}{2\alpha} \tilde{\rho} \sigma^\frac{1-\alpha}{2\alpha})^\alpha] + \varepsilon^\alpha \text{Tr}[(\sigma^\frac{1-\alpha}{2\alpha} P \sigma^\frac{1-\alpha}{2\alpha})^\alpha] \\
\label{Impl3}
\implies \quad & \text{Tr}[(\sigma^\frac{1-\alpha}{2\alpha} \rho \sigma^\frac{1-\alpha}{2\alpha})^\alpha] \leq \text{Tr}[(\sigma^\frac{1-\alpha}{2\alpha} \tilde{\rho} \sigma^\frac{1-\alpha}{2\alpha})^\alpha] + \varepsilon^\alpha \, . 
\end{align}
where in~\eqref{Impl1} we used that the trace functional $M \rightarrow \text{Tr}(f(M))$ inherits the monotonicity from $f$~(see e.g.~\cite{carlen2010trace}) and in~\eqref{Impl2} we used that for two positive semidefinite matrices $P$ and $Q$ and $\alpha \in (0,1)$ it holds  $\text{Tr}((P + Q)^\alpha) \leq \text{Tr}(P^\alpha) + \text{Tr}(Q^\alpha)$~\cite{bhatia1997matrix,marwah2022uniform}. The last implication~\eqref{Impl3} follows from the inequality $\text{Tr}[(\sigma^\frac{1-\alpha}{2\alpha} P \sigma^\frac{1-\alpha}{2\alpha})^\alpha] \leq 1$ for $\alpha \in (0,1)$. Therefore we obtain
\begin{equation}
\label{basic continuity}
\widetilde{Q}_\alpha(\rho \|\sigma) \leq \widetilde{Q}_\alpha(\tilde{\rho} \|\sigma) + \varepsilon^\alpha 
\end{equation}
The above relation holds also if we exchange $\rho$ and $\tilde{\rho}$. We now consider separately the two cases $\widetilde{D}_\alpha(\rho \| \sigma) > \widetilde{D}_\alpha(\tilde{\rho} \| \sigma)$ and $\widetilde{D}_\alpha(\rho \| \sigma) < \widetilde{D}_\alpha(\tilde{\rho} \| \sigma)$. For  $\widetilde{D}_\alpha(\rho \| \sigma) > \widetilde{D}_\alpha(\tilde{\rho} \| \sigma)$ we use that $\widetilde{Q}_\alpha(\tilde{\rho} \|\sigma) \leq \widetilde{Q}_\alpha(\rho \|\sigma) + \varepsilon^\alpha$ and we obtain for $\varepsilon \leq \widetilde{Q}_\alpha(\rho \|\sigma)^\frac{1}{\alpha}$
\begin{align}
\widetilde{D}_\alpha(\rho \| \sigma)-\widetilde{D}_\alpha(\tilde{\rho} \| \sigma) &=  - \frac{1}{\alpha-1} \log{\left(\frac{\widetilde{Q}_\alpha(\tilde{\rho} \| \sigma)}{\widetilde{Q}_\alpha(\rho \| \sigma)}\right)} \\
 &\leq - \frac{1}{\alpha-1} \log{\left(1 + \frac{\varepsilon^\alpha}{\widetilde{Q}_\alpha(\rho \| \sigma)}\right)} \\
&\leq \frac{1}{\alpha-1} \log{\left(1 - \frac{\varepsilon^\alpha}{\widetilde{Q}_\alpha(\rho \| \sigma)}\right)} \, .
\end{align}
where in the last inequality we used that $\log{(1+x)} \leq -\log{(1-x)}$ for any $0 \leq x \leq 1$. 
We then take the absolute value and for $\widetilde{D}_\alpha(\rho \| \sigma) > \widetilde{D}_\alpha(\tilde{\rho} \| \sigma)$ we get the bound
\begin{equation}
|\widetilde{D}_\alpha(\rho \| \sigma)-\widetilde{D}_\alpha(\tilde{\rho} \| \sigma)| \leq \frac{1}{\alpha-1} \log{\left(1 - \frac{\varepsilon^\alpha}{\widetilde{Q}_\alpha(\rho \| \sigma)}\right)}
\end{equation}
Instead, in the case  $\widetilde{D}_\alpha(\rho \| \sigma) < \widetilde{D}_\alpha(\tilde{\rho} \| \sigma)$ we use that $\widetilde{Q}_\alpha(\tilde{\rho} \|\sigma) \geq \widetilde{Q}_\alpha(\rho \|\sigma) - \varepsilon^\alpha$ and for $\varepsilon \leq \widetilde{Q}_\alpha(\rho \|\sigma)^\frac{1}{\alpha}$ we get
\begin{align}
\widetilde{D}_\alpha(\rho \| \sigma)-\widetilde{D}_\alpha(\tilde{\rho} \| \sigma) &=  - \frac{1}{\alpha-1} \log{\left(\frac{\widetilde{Q}_\alpha(\tilde{\rho} \| \sigma)}{\widetilde{Q}_\alpha(\rho \| \sigma)}\right)} \\
 &\geq - \frac{1}{\alpha-1} \log{\left(1 - \frac{\varepsilon^\alpha}{\widetilde{Q}_\alpha(\rho \| \sigma)}\right)} \, .
\end{align}
We then take the absolute value of the previous expression and we get for $\widetilde{D}_\alpha(\rho \| \sigma) < \widetilde{D}_\alpha(\tilde{\rho} \| \sigma)$ the same bound
\begin{equation}
|\widetilde{D}_\alpha(\rho \| \sigma)-\widetilde{D}_\alpha(\tilde{\rho} \| \sigma)| \leq \frac{1}{\alpha-1} \log{\left(1 - \frac{\varepsilon^\alpha}{\widetilde{Q}_\alpha(\rho \| \sigma)}\right)} \, .
\end{equation} 
\end{proof}

The previous proposition implies that also the related resource monotones are continuous. Indeed,
\begin{corollary}
\label{Continuity monotones}
Let $\alpha \in (0,1)$ and $\rho \in \mathcal{S}_\bullet(\mathcal{H})$. Then for any $\tilde{\rho} \in \mathcal{S}_\bullet(\mathcal{H})$ such that $\Delta(\rho, \tilde{\rho}) \leq \varepsilon \leq \mathcal{Q}_\alpha(\rho)^\frac{1}{\alpha}$ we have  
\begin{equation}
\label{continuity monotones_eq}
|\mathfrak{D}_\alpha(\rho)-\mathfrak{D}_\alpha(\tilde{\rho})| \leq \frac{1}{\alpha-1} \log{\left(1 - \frac{\varepsilon^\alpha}{\mathcal{Q}_\alpha(\rho) } \right)} \, .
\end{equation}
\end{corollary}
\begin{proof}
We have
\begin{align}
&\mathcal{Q}_\alpha(\tilde{\rho}) = \widetilde{Q}_\alpha(\tilde{\rho} \|  \sigma^*_{\tilde{\rho}}) \leq \widetilde{Q}_\alpha(\rho \|  \sigma^*_{\tilde{\rho}}) +\varepsilon^\alpha \leq \mathcal{Q}_\alpha(\rho)  + \varepsilon^\alpha \, ,
\end{align}
where we used inequality~\eqref{basic continuity} and introduced the optimiser $\sigma^*_{\tilde{\rho}}$ of $\mathcal{Q}_\alpha(\tilde{\rho})$. Note that the above inequality is analogous to~\eqref{basic continuity} and holds also if we exchange $\rho$ and $\tilde{\rho}$.  We therefore follow the same steps of Proposition~\ref{continuity bound} and we obtain~\eqref{continuity monotones_eq}.
\end{proof}

\section{Auxiliary results}
In the main text we introduced correlated catalytic transformations.
In~\cite{Sagawa} the authors introduced two different parameters for both the error on the output state of the system and the correlations between the system and the catalyst after the transformations. We show that the two definitions are equivalent, indeed
\begin{lemma}
\label{equivalence}
The following two statements are equivalent
\begin{enumerate}[ {(}1{)} ]
\item For any $\varepsilon > 0$ there exists a free operation $\mathcal{N}$ and a catalyst state $\nu$ such that $\mathcal{N}(\rho \otimes \nu) = \tau$, $\textup{Tr}_{\mathcal{H}}[\tau] = \nu$ and $P(\rho' \otimes \nu, \tau)< \varepsilon$.
\item For any $\varepsilon, \delta > 0$ there exists a free operation $\mathcal{N}$ and a catalyst state $\nu$ such that $\mathcal{N}(\rho \otimes \nu) = \tau$, $\textup{Tr}_{\mathcal{H}}[\tau] = \nu$, $d(\rho'', \rho')< \varepsilon$ and $D(\tau \| \rho'' \otimes \nu) < \delta$ where $\rho'' = \textup{Tr}_{\mathcal{H'}}[\tau]$.
\end{enumerate}
\end{lemma}

\begin{proof}
We first prove that $\mathit{(1)} \!\! \implies \!\! \mathit{(2)}$. The bound on the trace distance $d(\rho'', \rho')$ follows from data-processing and the relationship $d(\rho, \sigma) \leq P(\rho, \sigma)$ between trace distance and purified distance that holds for any two quantum states $\rho$ and $\sigma$. We get
\begin{align}
d(\rho', \rho'') \leq d(\rho' \otimes \nu, \tau) \leq P(\rho' \otimes \nu, \tau) < \varepsilon\,. 
\end{align}
To bound the correlations we first bound the trace distance $d(\rho'' \otimes \nu, \tau)$ using triangular inequality, namely
\begin{align}
d(\rho'' \otimes \nu, \tau) \leq d(\rho'' \otimes \nu, \rho' \otimes \nu) + d(\rho' \otimes \nu, \tau) < 2 \varepsilon \,.
\end{align}
Then we can bound the mutual information using continuity of the conditional entropy. 
The mutual information between the system S and the catalyst C in the state $\tau$ is defined as $I(S:C)_{\tau} := D(\tau \| \rho'' \otimes \nu)$. We can rewrite
\begin{align}
I(S:C)_{\tau}=H(S)_{\tau}-H(S|C)_{\tau} \, , 
\end{align}
where $H(S)$ is the von Neumann entropy of the system $S$ and $H(S|C)$ is the conditional entropy of $S$ given $C$.
The conditional entropy is continuous. In particular, if $\| \rho-\sigma\|_1 \leq \varepsilon$ then it follows $|H(S|C)_{\rho}-H(S|C)_{\sigma}| \leq 4 \varepsilon \log{d_{\mathcal{H}}} +2 h(\varepsilon)$ where $d_{\mathcal{H}}$ is the dimension of the system $S$~\cite{Winter_continuity}. 
Since conditioning reduces the entropy $H(S|C)_{\rho'' \otimes \nu}= H(S)_{\rho''} \geq H(S|C)_{\tau}$ we obtain
\begin{align}
I(S:C)_{\tau} &\leq H(S)_{\tau}-(H(S|C)_{\rho'' \otimes \nu}-16 \varepsilon\log{d_\mathcal{H}}-2h(4\varepsilon))\\
&= 16 \varepsilon \log{d_{\mathcal{H}}} + 2h(4\varepsilon)\, , 
\end{align}
since $H(S)_{\tau}=H(S|C)_{\rho'' \otimes \nu}$ have the same marginal. 

For the reverse implication $\mathit{(2)} \!\!  \implies \!\! \mathit{(1)}$ we use that $P(\rho, \sigma) \leq \sqrt{2d(\rho, \sigma)}$, triangular inequality and quantum Pinsker's inequality $D(\rho \| \sigma) \geq \frac{1}{2 \ln{2}} \| \rho -\sigma \|_1^2$~(see ,e.g.,[\citenum{Watrous}, Theorem 5.15]). We find
\begin{align}
P(\rho' \otimes \nu, \tau) &\leq \sqrt{2 d(\rho' \otimes \nu, \tau)} 
\leq \sqrt{2 (d(\rho' \otimes \nu, \rho'' \otimes \nu) + d(\rho'' \otimes \nu, \tau)) }
 < \sqrt{2\left( \varepsilon + \sqrt{\frac{\ln{2}}{2} \delta}\right)}\,.
\end{align}
\end{proof}
We remark that our definition differs from the one given in~\cite{Wilming} where high correlations in the output state between the system and the catalyst are still allowed. 

In general catalytic state transformation between any two states $\rho$ and $\rho'$ is not possible. 
Indeed, it is a straightforward fact that (see also \cite{Brandao,Sagawa,Takagi})

\begin{lemma}
\label{Basic lemma}
Let $\rho$ and $\rho'$ be two states and $\mathfrak{R}$ a resource monotone. If any of the following statements hold
\begin{align}
&\mathit{(1)} \,\,\, \rho' \in \setFO(\rho) \notag\\
&\mathit{(2)} \,\,\, \rho' \in \overline{\setFO(\rho)} \,\,\, \text{and} \,\,\, \mathfrak{R} \,\, \text{is lower semicontinuous} \notag\\
&\mathit{(3)} \,\,\, \rho' \in \setCO(\rho) \,\,\, \text{and} \,\,\, \mathfrak{R} \,\, \text{is tensor product additive} \notag\\
&\mathit{(4)} \,\,\, \rho' \in \overline{\setCO(\rho)}  \, \, \,  \text{and} \,\,\, \mathfrak{R} \,\,  \text{is tensor product additive and lower semicontinuous}\notag \\
&\mathit{(5)} \,\,\, \rho' \in \setCCO(\rho) \, \, \, \text{and} \,\,\, \mathfrak{R} \,\,  \text{is superadditive and tensor product additive} \notag\\
&\mathit{(6)} \,\,\, \rho' \in \overline{\setCCO(\rho)} \, \, \, \text{and} \,\,\, \mathfrak{R} \text{ is superadditive, tensor product additive and lower semicontinuous} \notag
\end{align}
than we must have $\mathfrak{R}(\rho) \geq \mathfrak{R}(\rho')$.
\end{lemma}

\begin{proof}
The statements $\mathit{(1)}$ and $\mathit{(2)}$ are trivial. We prove only $\mathit{(4)}$ and $\mathit{(6)}$ since the proofs for $\mathit{(3)}$ and $\mathit{(5)}$ follow similarly. 
If $\rho' \in \overline{\setCO(\rho)}$, using tensor product additivity and monotonicity under free operations
\begin{align}
 \mathfrak{R}(\rho) + \mathfrak{R}( \nu)  = \mathfrak{R}(\rho \otimes \nu) \geq \mathfrak{R}(\rho'_\varepsilon \otimes \nu) =  \mathfrak{R}(\rho'_\varepsilon) + \mathfrak{R}(\nu) \, . 
\end{align}
where $\rho'_\varepsilon$ is a state $\varepsilon$-close to $\rho'$. 
Since $\rho'$ and $\rho'_\varepsilon$ are arbitrarily close and $\mathfrak{R}$ is lower semicontinuous, the above relation implies $\mathfrak{R}(\rho) \geq \mathfrak{R}(\rho')$.

If $\rho' \in \overline{\setCCO(\rho)}$, using tensor product additivity, monotonicity under free operations and superadditivity 
\begin{align}
 \mathfrak{R}(\rho) + \mathfrak{R}( \nu)  = \mathfrak{R}(\rho \otimes \nu) \geq \mathfrak{R}(\tau) \geq  \mathfrak{R}(\text{Tr}_C[\tau]) + \mathfrak{R}(\nu) \, . 
\end{align}
Since $\rho'$ and $\text{Tr}_\mathcal{H'}[\tau]$ are arbitrarily close and $\mathfrak{R}$ is lower semicontinuous, the above relation implies $\mathfrak{R}(\rho) \geq \mathfrak{R}(\rho')$. 
\end{proof}
Note that for the proof to hold in the cases $\mathit{(5)}$ and $\mathit{(6)}$ we do not need to assume that $\rho'$ can be achieved with arbitrarily small correlations. 

\bigskip
\textit{Remark} If $\mathfrak{R} =\mathfrak{D}_\alpha$ with $\alpha \in [1/2,1)$, if $\rho' \in \overline{\setCO(\rho)}$ and $\mathfrak{D}_\alpha$ is additive for the state $\rho'$, then we must have $\mathfrak{D}_\alpha(\rho) \geq \mathfrak{D}_\alpha(\rho')$. Indeed, following the same steps of the proof of Proposition~\ref{continuity bound}, if $d(\rho, \tilde{\rho}) \leq \varepsilon$, we obtain
\begin{align}
\widetilde{Q}_\alpha(\rho \otimes \nu \| \sigma) \leq \widetilde{Q}_\alpha (\tilde{\rho} \otimes \nu \| \sigma) + \varepsilon^\alpha \widetilde{Q}_\alpha(P'\otimes \nu \| \sigma) \leq \widetilde{Q}_\alpha (\tilde{\rho} \otimes \nu \| \sigma) + \varepsilon^\alpha \widetilde{Q}_\alpha(\nu \| \text{Tr}_{\mathcal{H}}(\sigma)) 
\end{align}
where in the last inequality we used the data-processing inequality under partial trace. This implies that $|\mathfrak{D}_\alpha(\tilde{\rho} \otimes \nu) - \mathfrak{D}_\alpha(\rho \otimes \nu)| \leq \frac{1}{\alpha-1} \log{\left(1 - \frac{\varepsilon^\alpha}{\mathcal{Q}_\alpha(\rho)} \right)}$. Then the chain of inequalities
\begin{align}
\mathfrak{D}_\alpha(\rho) + \mathfrak{D}_\alpha(\nu) & \geq  \mathfrak{D}_\alpha(\rho \otimes \nu)\\
& \geq \mathfrak{D}_\alpha(\rho'_{\epsilon} \otimes \nu) \\
& \geq \mathfrak{D}_\alpha(\rho' \otimes \nu) - \frac{1}{\alpha-1}\log{\left(1-\frac{\varepsilon^\alpha}{\mathcal{Q}_\alpha(\rho')}\right)} \\
&= \mathfrak{D}_\alpha(\rho') + \mathfrak{D}_\alpha(\nu) - \frac{1}{\alpha-1}\log{\left(1-\frac{\varepsilon^\alpha}{\mathcal{Q}_\alpha(\rho')}\right)} \, ,
\end{align}
which holds for any $\varepsilon$, implies that $\mathfrak{D}_\alpha(\rho) \geq \mathfrak{D}_\alpha(\rho')$. 

\section{Proof of the main Theorem}
\label{Proof of the main Theorem}
In this section we provide the proof of our main theorem. In the following we say that $ \mathfrak{D}_\alpha$ is additive for the state $\rho$ if $\mathfrak{D}_\alpha(\rho \otimes \nu) = \mathfrak{D}_\alpha(\rho) + \mathfrak{D}_\alpha(\nu)$ for any state $\nu$ of the  catalyst.
\begin{theorem}
\label{fidelity theorem 1}
Assume that $\rho, \rho' \in \mathcal{S}(\mathcal{H})$ and $\alpha \in [1/2,1)$ such that $\mathfrak{D}_\alpha$ is additive for the state $\rho'$ and $\mathfrak{D}_\alpha(\rho) < \mathfrak{D}_\alpha(\rho')$. Then, for any $\varepsilon$-correlated catalytic transformation with catalyst $\nu$ mapping $\rho$ into $\rho'$, we have
\begin{align}
\mathcal{Q}_\alpha(\nu) \leq \frac{\varepsilon^\alpha}{\mathcal{Q}_\alpha(\rho)-\mathcal{Q}_\alpha(\rho')}\, . 
\end{align}
Moreover we get that $\mathfrak{D}(\nu) = \Omega\left( \log{\frac{1}{\varepsilon}}\right)$ and $L \mathfrak{R}_g(\nu) = \Omega\left( \log{\frac{1}{\varepsilon}}\right)$.
\end{theorem}
We point out that the above theorem actually holds also if we lift the restriction $\text{Tr}_{\mathcal{H}}[\tau] = \nu$ and hence we do not need the catalyst to be exactly recovered after the transformation.

\begin{figure}
\centering
{\includegraphics[width=.65\textwidth]{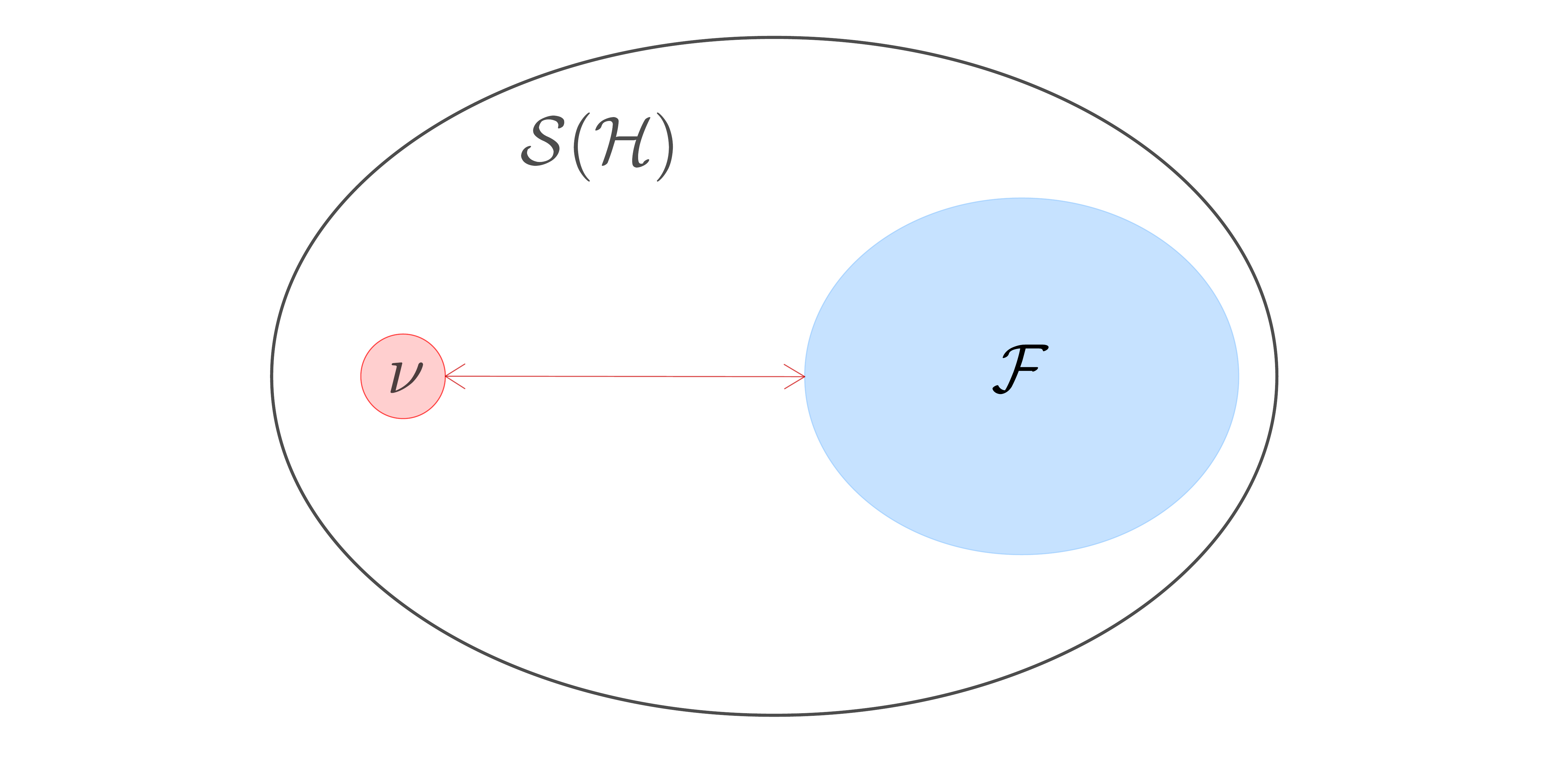}}
\caption{Intuitive geometric representation for the main result states in Theorem~\ref{fidelity theorem 1} . The distance between the catalyst and the free set must increase as the error in the correlated catalytic transformation goes to zero}
\label{fig: free_set 1}
\end{figure}

We are now ready to prove Theorem~\ref{fidelity theorem 1}. 
\begin{proof}
We have the following chain of inequalities 
\begin{align}
\mathfrak{D}_{\alpha}(\rho) + \mathfrak{D}_\alpha(\nu) \geq \mathfrak{D}_{\alpha}(\rho \otimes \nu)+f = \mathfrak{D}^\varepsilon_{\alpha}(\rho \otimes \nu) \geq \mathfrak{D}^\varepsilon_{\alpha}(\tau) \geq \mathfrak{D}_{\alpha}(\rho'  \otimes \nu) \,,
\end{align}
where $f :=  \mathfrak{D}^\varepsilon_{\alpha}(\rho \otimes \nu) - \mathfrak{D}_{\alpha}(\rho \otimes \nu)$. The first inequality follows the subadditivity of $\mathfrak{D}_\alpha$, the second inequality from Corollary \ref{data-processing} and the third one is a consequence of how we chose the smoothing in~\eqref{smoothed} for $\alpha \in [1/2,1)$. The inequality chain, together with the additivity assumption for $\rho'$, gives $f \geq \Delta \mathfrak{D}_\alpha := \mathfrak{D}_\alpha(\rho')-\mathfrak{D}_\alpha(\rho)$. 

Therefore we get
\begin{align}
\Delta \mathfrak{D}_\alpha \leq f = \mathfrak{D}^\varepsilon_{\alpha}(\rho \otimes \nu) - \mathfrak{D}_{\alpha}(\rho \otimes \nu) &\leq \frac{1}{\alpha-1}\log{\left( 1 - \frac{\varepsilon^\alpha}{\mathcal{Q}_\alpha(\rho \otimes \nu)}\right)}  \\
&\leq \frac{1}{\alpha-1}\log{\left( 1 - \frac{\varepsilon^\alpha}{\mathcal{Q}_\alpha(\rho) \mathcal{Q}_\alpha(\nu)}\right)} \, . 
\end{align}
where we used Corollary~\ref{Continuity monotones} and that the optimiser $\eta$ of $\mathfrak{D}^\varepsilon_{\alpha}(\rho \otimes \nu):= \mathfrak{D}_{\alpha}(\eta) $ satisfies $\Delta(\eta, \rho \otimes \nu) \leq P(\eta, \rho \otimes \nu) \leq \varepsilon$. The last inequality follows from $\mathcal{Q}_\alpha(\rho \otimes \nu ) \geq \mathcal{Q}_\alpha(\rho) \mathcal{Q}(\nu)$.
 
Inverting the above relation it follows that
\begin{equation}
\mathcal{Q}_\alpha(\nu) \leq \frac{ \varepsilon^\alpha}{\mathcal{Q}_\alpha(\rho)-\mathcal{Q}_\alpha(\rho')} \, .
\end{equation}
Note that in the domain for which the argument of the logarithm is negative, i.e. $\varepsilon^\alpha \geq \mathcal{Q}_\alpha(\rho)\mathcal{Q}(\nu)$, the main theorem already follows immediately. Finally, because of the inequality 
\begin{align}
L\mathfrak{R}_g(\rho) = \mathfrak{D}_{\max}(\rho)  \geq  \mathfrak{D}(\rho) \geq \mathfrak{D}_{\alpha}(\rho) \,,
\end{align}
which holds for any state $\rho$ and for $\alpha \in [1/2,1)$, we also obtain that the generalised robustness of the catalyst state (and hence the robustness) and the monotone $\mathfrak{D}$ must both diverge as the the error vanishes. In particular we get that $\mathfrak{D}(\nu) = \Omega\left( \log{\frac{1}{\varepsilon}}\right)$ and $L \mathfrak{R}_g(\nu) = \Omega\left( \log{\frac{1}{\varepsilon}}\right)$.
\end{proof}

Fig.~\ref{fig: free_set 1} shows the geometric content of~Theorem \ref{fidelity theorem 1}. As the error of the transformation $\varepsilon$ approaches zero, the distance between the catalyst state $\nu$ and the set of free states $\mathcal{F}$ must increase, meaning that we need to prepare a more resourceful catalyst. In the limit $\varepsilon=0$ the catalyst state $\nu$ must be orthogonal to the set of free states, and hence to perform the transformation we would need an infinitely resourceful catalyst \footnote{We point out that Theorem 3 as stated in~\cite{Sagawa} is incorrect since if we could choose any two catalyst states, then we could just select two orthogonal states and perform a measurement and prepare map conditioned on the state of the catalyst and thus accomplish an exact transformation between any two pairs of states. The theorem has been fixed by the authors in an Erratum by adding an additional condition on the support of the catalysts.}.

In Appendix~\ref{tighter} we prove a tighter bound for the $\alpha=1/2$ case and in the next sections we discuss the consequences of Theorem~\ref{fidelity theorem 1} for resource theory of athermality, entanglement theory and resource theory of coherence.

\medskip

\textit{Remark} $\,$ Note that in classical resource theory of athermality $\widetilde{D}_\alpha(\gamma \| p)$ for $\alpha \in (1/2,1)$ is equal to $\frac{\alpha}{\alpha-1}\widetilde{D}_{1-\alpha}(p \| \gamma)$ where $1-\alpha \in (0,1/2)$. By repating the same steps above, the latter quantity, if smoothed, still satisfies the data-procesing inequality and therefore the main theorem holds also for $\widetilde{D}_\alpha$ in the range $\alpha \in (1/2,1)$ with the arguments exchanged.

\section{Resource theory of athermality}
\label{Resource theory of athermality}
In resource theory of athermality the Gibbs state $\gamma = e^{-\beta H} /Z$ is the only free state in the theory, where $\beta$ is the inverse temperature, $H$ the Hamiltonian of the system and $Z$ is the normalisation factor (partition function).
We remark that our results apply to both resource theory of athermality with thermal operations and resource theory of athermality under Gibbs preserving maps since in both resource theories free operations keep the Gibbs state invariant~\cite{faist2015gibbs,lostaglio2018elementary}. However since for the former we do not know the resource monotone characterizing the set $\overline{\setCCO}$ in the quantum case where the two theories are different, for our considerations we will mainly focus on the latter where the relevant resource monotone is the non-equilibrium free energy~\cite{muller2018correlating,Sagawa}.

Given that the system and the catalyst are non-interacting,   the total Gibbs state reduces to the tensor product $\gamma_{SC} = \gamma_S \otimes \gamma_C$. Therefore $ \widetilde{D}_\alpha(\rho \otimes \sigma \|  \gamma_{S} \otimes  \gamma_{C}) = \widetilde{D}_\alpha(\rho , \gamma_{S}) + \widetilde{D}_\alpha(\sigma, \gamma_C) $ is additive~\cite{Tomamichel} and hence Theorem~\ref{fidelity theorem 1} holds for this specific case. 

In~\cite{Sagawa} the authors established that $\rho$ is transformable into $\rho'$ by a correlated catalytic transformation if and only if the free energies are ordered, namely $ \mathfrak{D}(\rho)\geq  \mathfrak{D}(\rho')$.
Therefore, for the bound \eqref{fidelity theorem 1} to be meaningful, there must exist states that both satisfy $\mathfrak{D}_\alpha(\rho) < \mathfrak{D}_\alpha(\rho')$ for $\alpha \in [1/2,1)$ and $ \mathfrak{D}(\rho) \geq  \mathfrak{D}(\rho')$. In the following we find states that both satisfy $\mathfrak{D}(\rho) \geq \mathfrak{D}(\rho')$ and $\mathfrak{D}_{1/2}(\rho)<\mathfrak{D}_{1/2}(\rho')$ (or $\mathfrak{F}(\rho) > \mathfrak{F}(\rho')$). States with both these properties can be found in any dimensions. 
\subsection{Three-level system} 
In a three dimensional system it is possible to build classical states that both satisfy $\mathfrak{F}(\rho) > \mathfrak{F}(\rho')$ and $ \mathfrak{D}(\rho) \geq  \mathfrak{D}(\rho')$ and such that the fidelity gap $\sqrt{\mathfrak{F}(\rho)} - \sqrt{\mathfrak{F}(\rho')}$ is arbitrary close to one. Since in this section we consider only classical systems our analysis includes both resource theory of athermality with thermal operations and resource theory of athermality under Gibbs preserving maps. 
We find 
\begin{lemma}
For any $\delta > 0$, there exists a 3-level system with states $\rho, \rho' \in \mathcal{S}(\mathbb{C}^3)$ such that 
\begin{align}
 \mathfrak{D}(\rho) \geq  \mathfrak{D}(\rho') \quad \text{and} \quad \sqrt{\mathfrak{F}(\rho)} - \sqrt{\mathfrak{F}(\rho')} > 1-\delta
\end{align}
\end{lemma}

\begin{proof}
Let us consider a three level system with Hamiltonian $H = \sum_{i=1}^3 E_i |i\rangle\langle i|$ and the following two diagonal states in the Hamiltonian eigenbasis $\rho, \rho'$ (see Fig.~\ref{fig:qutrit states}(a))

\begin{align}
\label{qutrit}
\rho = \frac{1}{Z_{\rho}} \bigg[ 1,e^{-\beta E_{2}} , \mu \bigg] \qquad \text{and} \qquad \rho' = \frac{1}{Z_{\rho'}} \bigg[ 0, e^{-\beta E_{2}} , e^{-\beta E_{3}} \bigg]\, , 
\end{align}
where $Z_{\rho} = 1+e^{-\beta E_2} + \mu $ and $Z_{\rho'} = e^{-\beta E_2} + e^{-\beta E_3} $.

We set the energy scale such that $E_1=0$ (in some unit of measurement) and therefore $\gamma_1 = 1/Z$ where $\gamma_i = e^{-\beta E_i}/Z$.
We introduce also the embedding channel~\cite{Korzekwa}.
\begin{definition}[Embedding channel]
Given a thermal distribution with rational entries $\gamma_i=D_i/D$ with $D_i,D \in \mathbb{N}$, the embedding channel $\Gamma$ maps a $d$-dimensional probability distribution $p$ to a $D$-dimensional probability distribution $\hat{p}$ as follows 
\begin{align}
\hat{p} = \Gamma(p) =\bigg[\underbrace{ \frac{p_1}{D_1},  \dots , \frac{p_1}{D_1}}_{D_1\, \text{times}}, \dots \dots, \underbrace{\frac{p_d}{D_d},  \dots, \frac{p_d}{D_d}}_{D_d\, \text{times}}\bigg]\,.
\end{align}
\end{definition}

Let us choose the Hamiltonian of the system such that $D_3 = 1$ and therefore $\gamma_3 = 1/D$.
The embedding channel maps the states \eqref{qutrit} into
\begin{align}
 \hat{\rho}  = \bigg[\underbrace{\kappa ,  \dots \dots , \kappa}_{D-1\, \text{times}}, \mu' \bigg] \quad \text{and} \quad \hat{\rho}' = \bigg[0, \dots , 0 , \underbrace{ \kappa'  \dots ,\kappa'}_{D_2+1 \, \text{times}}\bigg]\,,
\end{align}
where $\kappa =  Z/(Z_{\rho} D)$ , $\mu' = \mu/Z_{\rho}$ and $\kappa' = Z/(Z_{\rho'}D)$. 

Since the embedding channel maps probability distributions into probability distributions, we can set $\kappa = \frac{1-\mu'}{D-1}$, $\kappa' =(1/(D-1))^{1-\varepsilon}$. Moreover we choose $D_2+1 = (D-1)^{1-\varepsilon}$ which we can always satisfy for any $\varepsilon>0$ with $D_2$ integer with arbitrary accuracy as $D \rightarrow \infty$.
Then we get the following classical states
\begin{align}
\label{classical states}
\hat{\rho} = \bigg[\underbrace{ \frac{1-\mu'}{D-1},  \dots \dots, \frac{1-\mu'}{D-1}}_{D-1\, \text{times}} ,\mu' \bigg]  \qquad \text{and} \qquad
\hat{\rho}' = \bigg[ 0,  \dots, 0, \underbrace{ \left( \frac{1}{D-1} \right)^{1-\varepsilon} \hskip-15pt ,\dots, \left( \frac{1}{D-1} \right)^{1-\varepsilon} }_{(D-1)^{1-\varepsilon}\, \text{times}}  \bigg]\,.
\end{align}
The embedding channel maps the Gibbs state into the fully mixed state $\eta_D = \mathds{1} /D$. Using that $\widetilde{D}_{\alpha}(\rho \| \eta_D) = -H_\alpha(\rho) + \log{D}$, where $H_{\alpha}$ are the $\alpha$-Renyi divergences $H_\alpha = \frac{1}{1-\alpha}\log{\text{Tr}(\rho^\alpha)}$, the conditions on the fidelity $\mathfrak{F}(\rho) > \mathfrak{F}(\rho') $ and the relative entropy $\mathfrak{D}(\rho) \geq \mathfrak{D}(\rho')$ turn into 
 \begin{align}
 h_{\text{bin}}(\mu')+(1-\mu')\log{(D-1)} = H(\hat{\rho} )& \leq H(\hat{\rho} ') = (1-\varepsilon)\log{(D-1)}\\
 2\log{(\sqrt{\mu'}+\sqrt{D-1}\sqrt{1-\mu'})}= H_{\frac{1}{2}}(\hat{\rho}) &> H_{\frac{1}{2}}(\hat{\rho}')  =  (1-\varepsilon)\log{(D-1)} \, . 
 \end{align}
 Since $ h_{\text{bin}}(\mu') \leq 1$ the first condition is satisfied whenever $\mu' \geq \varepsilon + 1/\log{(D-1)}
$. Let us then fix $ \mu' = \varepsilon + 1/\log{(D-1)}$.
 Since $H_{1/2}(\hat{\rho})>\log{(D-1)}+\log{(1-\mu')}$  the second condition is satisfied for $(D-1)^\varepsilon>1/(1-\mu')$ which is always satisfied for $D$ big enough. Then, noting that the embedding channel preserves the fidelity,  asymptotically ($D \rightarrow  \infty$) the fidelities in the original space behave as 
 \begin{align}
F(\rho, \gamma) \sim (1-\varepsilon) \quad \text{and} \quad F(\rho', \gamma) \sim (1/D)^{\varepsilon}\,.
 \end{align}
 Therefore choosing $\varepsilon$ small enough we can always find $D$ big enough such that $F(\hat{\rho}, \eta_D) $ is arbitrary close to $1$ and $F(\hat{\rho}', \eta_D)$ is arbitrarily close to zero.
The condition $ D_2+1 = (D-1)^{1-\varepsilon}$ can be written for $D \gg 1$ as
\begin{align}
\gamma_2 + \gamma_3  \sim \left( \frac{1}{D} \right) ^\varepsilon \, .
\end{align}
Since $\gamma_3= 1/D$ it follows that asymptotically $\gamma_2 \sim (1/D)^{\varepsilon}$. 
Then, since the Gibbs state is normalised, $ \gamma_1 = 1 -\gamma_2 - \gamma_3 \sim 1 $ and therefore $Z \sim 1$. 
Since $ \gamma_3 = e^{-\beta E_3}/Z = 1/D$ the dimension of the embedding space $D$ scales exponentially with the gap $\Delta E := E_3 $ as $D \sim e^{\beta \Delta E}$. 
We have $ E_2 \sim \varepsilon ((1/\beta) \log{D}) $ and $ E_3 = (1/\beta) \log{D} $ and hence $E_2 \sim \varepsilon E_3$. The situation is depicted in Fig.~\ref{fig:qutrit states}(b) where the energies are measure in scale $(1/\beta) \log{D}$.
Obviously we can also find classical states with these behaviour for $D>3$ as we can always ignore the other dimensions.

\begin{figure}
\vskip-30pt
\centering
\subfloat[c][]
{\includegraphics[width=.6\textwidth]{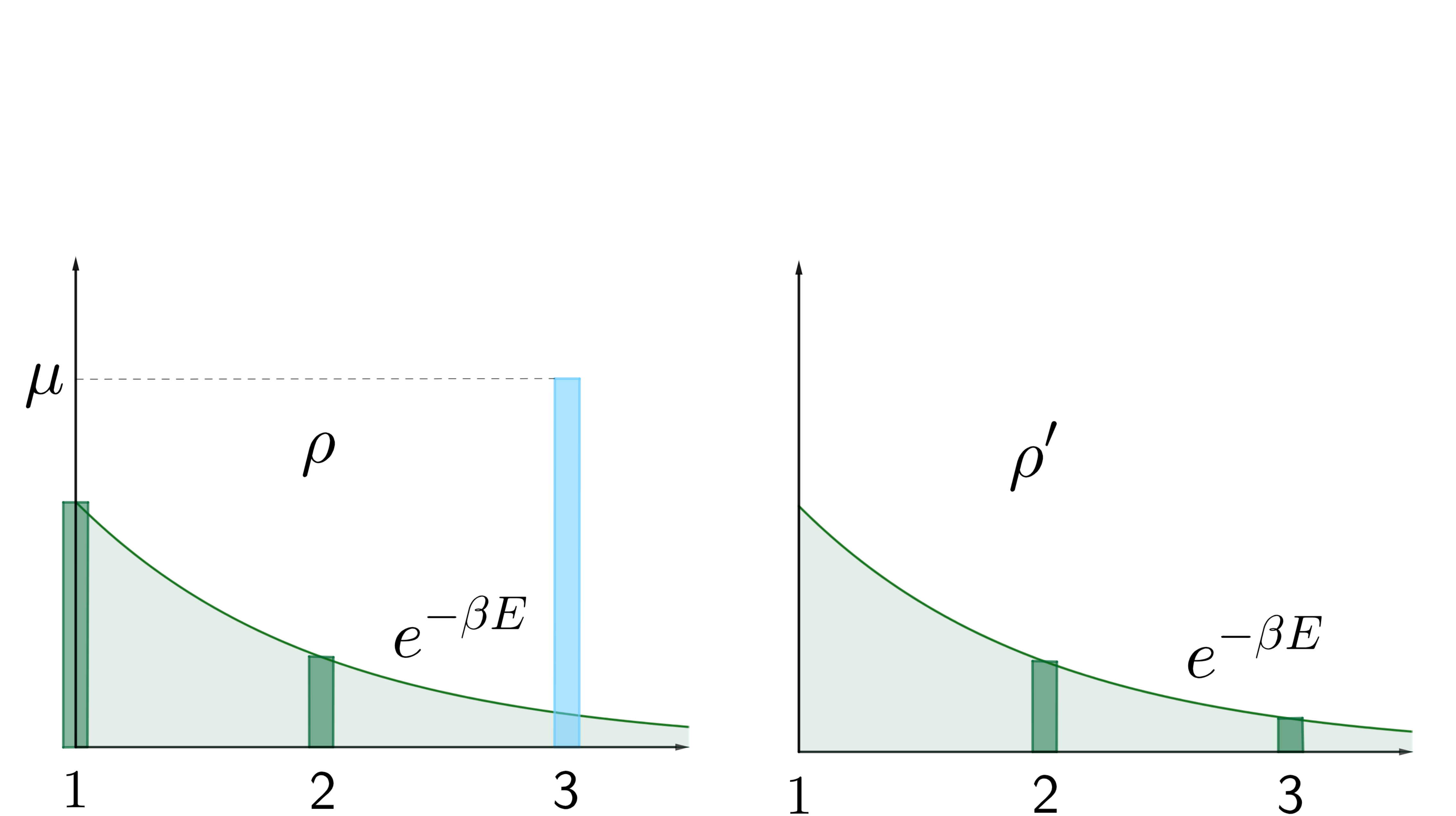}}
\subfloat[][]
{\includegraphics[width=.3\textwidth]{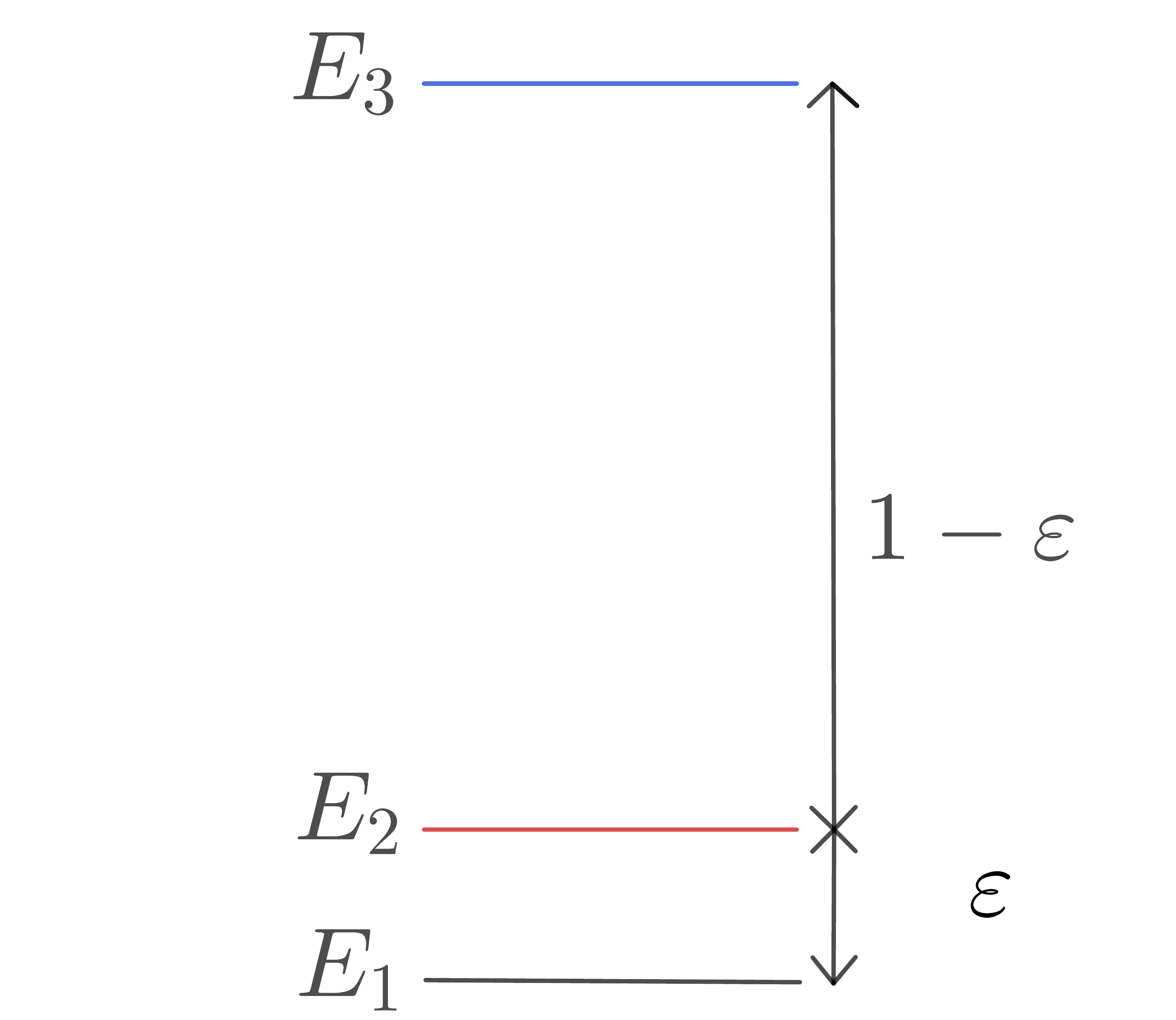}} 
\caption{(a) Shape of the unnormalised states $\rho$ and $\rho'$ in \eqref{qutrit}. The green decaying exponential is the unnormalised Gibbs state. (b) Energy levels of the 3-level system. The energies are measured in units $(1/\beta) \log{D}$.}
\label{fig:qutrit states}
\end{figure}
\end{proof}

\subsection{Qubit system}
Interestingly, even though classical states $\rho, \rho'$ that satisfy both $\mathfrak{D}(\rho) \geq \mathfrak{D}(\rho')$ and $\mathfrak{F}(\rho) > \mathfrak{F}(\rho')$ do not exist in two dimensions, quantum states that satisfy this requirement can be found. In the following we give a numerical example. 

We fix $\gamma = 0.999 |0 \rangle \langle 0| +0.001 |1 \rangle \langle 1|$ and look for such states numerically in the Bloch sphere. Fig.~\ref{qubit_states} shows the $x$-$z$ plane of the Bloch sphere where the different colors are associated with different ranges of $\mathfrak{D}(\rho)$. In the blue region $\mathfrak{D}(\rho)\leq 2$. The red and the green line correspond to the maximum and the minimum value that the fidelity gets along the line $\mathfrak{D}(\rho)=2$. We choose the input state $\rho$ and the output state $\rho'$ at the intersections between the free energy line $\mathfrak{D}(\rho)=2$ and the red and green fidelity lines, respectively. We find numerically that $\rho$ is the pure state at an angle $\theta \sim \pi/3.38$ with the vertical axis in the Bloch sphere.
The state $\rho'$ that maximize the fidelity gap along the constant relative entropy line $\mathfrak{D}(\rho)=2$ is approximately $\rho' \sim 0.713 |0 \rangle \langle 0| + 0.287 |1 \rangle \langle 1| $ and $\mathfrak{F}(\rho) - \mathfrak{F}(\rho') \sim 0.058$.
For what we discussed above, catalytic transformation  with vanishing error of these two 'hard-to-transform' states  would require an infinite free energy catalyst state $\nu$.
\begin{figure}
\centering
{\hskip-25pt \includegraphics[width=.45\textwidth]{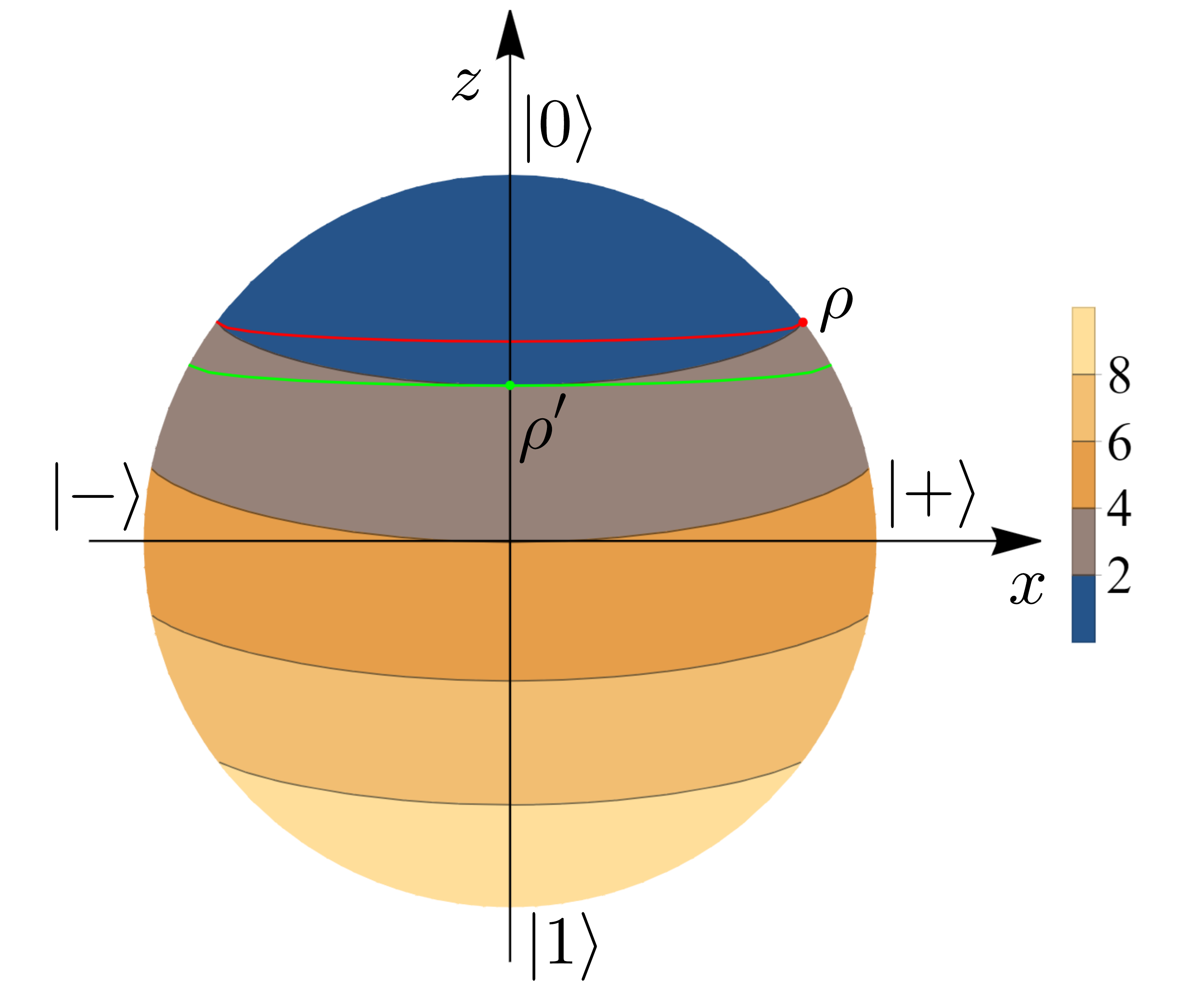}}
\caption{$x$-$z$ plane of the Bloch sphere. The different colors represent different relative entropy regions $\mathfrak{D}(\rho)$ for $\gamma = 0.999 |0 \rangle \langle 0| +0.001 |1 \rangle \langle 1|$. The red and the green line represent two different lines of constant fidelity. We choose our states $\rho, \rho'$ at the intersections between the red line and the green line with the relative entropy line $\mathfrak{D}(\rho) =2$, respectively.}
\label{qubit_states}
\end{figure}
\bigskip

From  Theorem $2$ we can immediately obtain some bounds noting by that $\mathfrak{F}(\nu) \geq F(\nu, (\min_i p_i) \mathds{1})$ $ \geq e^{-\beta \Delta E}/Z$.  We can therefore bound the dimension of the catalyst $d_\nu = \Omega(1/\varepsilon)$ if we keep the gap and the temperature constant or, if we fix the dimension of the catalyst, we get $\Delta E/(kT) = \Omega(\log{1/\varepsilon})$.
However, we identify the non-equilibrium free energy as the relevant physical quantity to be calculated for the catalyst in this setting. From Theorem~\ref{fidelity theorem 1} we get 
\begin{align}
\label{lower bound free energy}
\mathfrak{D}(\nu)= \Omega\left( \log{\frac{1}{\varepsilon}}\right) 
\, . 
\end{align}
Therefore, correlated catalytic transformation between any two states would require preparing a catalyst with an infinite amount of free energy as the error vanishes. 
However, we point out that the divergence is logarithmic and therefore does not rule out the possibility of achieving very small errors in the transformation.

\subsection{A procedure for correlated catalytic conversion}
\label{A procedure for correlated catalytic conversion}
In this section we prove that the catalyst first introduced in~\cite{duan2005multiple} and recently discussed in~\cite{Sagawa} is optimal in the sense defined above in the resource theory of athermality.
The catalyst state $\nu$ and the Gibbs state $\gamma$ are given as~\cite{Sagawa,duan2005multiple} 
\begin{align}
\label{catalyst state}
\nu = \frac{1}{n}\sum_{k=1}^n \rho^{\otimes k-1}\otimes \Xi_{n-k}\otimes |k \rangle \langle k| \qquad \gamma=\frac{1}{n}\sum_{k=1}^n \eta^{\otimes k-1}\otimes \eta'^{\otimes n-k}\otimes |k \rangle \langle k| \, .
\end{align}
We further obtain
\begin{align}
\mathfrak{D}(\nu) = D(\nu \| \gamma) & \leq \frac{1}{n} \sum_{k=1}^n D(\rho^{\otimes k-1}\otimes \Xi_{n-k} \| \eta^{\otimes k-1}\otimes \eta'^{\otimes n-k}) \\
& \leq  \frac{1}{n} \sum_{k=1}^n [(k-1) D(\rho \| \eta) + n D(\rho \| \eta)] \\
\label{free energy catalyst}
& \leq 2n D(\rho \| \eta) \, , 
\end{align}
where we used joint convexity, additivity under tensor products and data-processing of the trace distance. Theorem 1 in~\cite{Sagawa} could also be formulated using the purified distance instead of the trace distance.  
If we set $P(\Xi, \rho'^{\otimes n}) \leq \varepsilon$, using the explicit form of the catalyst state $\nu$~\eqref{catalyst state} and the output state $\tau =  \frac{1}{n} \sum_{k=1}^n \rho^{\otimes k-1}\otimes \Xi_{n-k+1} \otimes |k \rangle \langle k|$ we obtain
\begin{align}
P(\tau, \rho'\otimes \nu)  & \leq \max \limits_{k} P(\Xi_{n-k+1} , \Xi_{n-k} \otimes \rho') 
\leq \max \limits_{k} P(\Xi_{n-k+1}, \rho'^{n-k+1})+P(\rho'^{n-k+1}, \Xi_{n-k} \otimes \rho') \leq 2 \varepsilon \, , 
\end{align}
where we used joint quasi-convexity, triangular
inequality and monotonicity under partial trace of the purified distance~\cite{Tomamichel}. 
We have that  $n \leq \frac{1}{\gamma}\log{\frac{1}{\varepsilon}}$ for some constant $\gamma$ \cite{Buscemi}. In appendix~\ref{error exponent} we find a qualitatively give a lower bound for the error exponent $\gamma$ for small entropy gaps. We obtain from~\eqref{free energy catalyst} the following behaviour for the free energy
\begin{align}
\mathfrak{D}(\nu) = O\left( \log{\frac{1}{\varepsilon}} \right) \, .
\end{align} 
By comparing this result with the lower bound \eqref{lower bound free energy} we establish that the procedure is optimal in $\varepsilon$. 

\section{Entanglement theory}
\label{Entanglement theory}

Let $\mathcal{H}_1 \otimes ... \otimes \mathcal{H}_m$ a multipartite Hilbert space. We call a state $\sigma \in \mathcal{S}(\mathcal{H}_1 \otimes ... \otimes \mathcal{H}_m)$ separable if it is of the form $\sigma = \sum_i p_i \sigma_i^1 \otimes ... \otimes \sigma_i^m $ for some local states $\sigma_j^k \in \mathcal{S}(\mathcal{H}_k)$ and a probability distribution $\{ p_i \}$. We denote the set of all separable states (free states) as usual by $\mathcal{F}$. In the following we consider input and output bipartite pure states $|\psi _{AB}\rangle, |\psi' _{AB}\rangle  \in \mathcal{S}(\mathcal{H}_A \otimes \mathcal{H}_B)$. 

We consider the most general situation in which the two parties involved in the protocol hold a mixed catalyst state. 
We can apply Theorem~\ref{fidelity theorem 1} since the monotones $\mathfrak{D}_\alpha$ are additive when one state is pure~\cite{Unpublished}. The squashed entanglement is an entanglement monotone that is superadditive, additive under tensor products and continuous~\cite{Christandl, Alicki}. Therefore, for what we have already mentioned, the squashed entanglement must decrease under LOCC operations~\cite{Kondra}.  For a bipartite pure state $|\psi _{AB}\rangle$, both the squashed entanglement and the relative entropy of entanglement $\mathfrak{D}$ reduce to the entanglement entropy \cite{vedral1998entanglement,Vlatko,horodecki2009quantum,bennett1996concentrating}, namely $E_{sq}(|\psi _{AB}\rangle) = \mathfrak{D}(|\psi _{AB}\rangle) = H(\psi_A)$ where $H(\rho) = -\text{Tr}(\rho \log{\rho})$ and $\psi_A = \text{Tr}_B{|\psi \rangle \! \langle \psi|_{AB}}$. Moreover, the condition  $H(\psi_A) \geq H(\psi'_A)$ is also a sufficient condition for approximated asymptotic, and hence correlated catalytic, transformation~\cite{bennett1996concentrating, Bennet2}. It then follows that $|\psi _{AB}\rangle$ can be catalytically transformed into $|\psi' _{AB}\rangle$ if and only if $\mathfrak{D}(|\psi _{AB}\rangle) \geq \mathfrak{D}(|\psi'_{AB}\rangle) $~\cite{Kondra}. We then look for states satisfying both $\mathfrak{F}(|\psi _{AB}\rangle) > \mathfrak{F}(|\psi' _{AB}\rangle)$ and $\mathfrak{D}(|\psi _{AB}\rangle) \geq \mathfrak{D}(|\psi'_{AB}\rangle) $. As in the previous case, states of this kind exist and the fidelity gap can be chosen arbitrarily close to $1$. Indeed we find 
\begin{lemma}
For any $\delta > 0$ there exist $d > 0$ and two states $|\psi _{AB}\rangle, |\psi' _{AB}\rangle \in \mathcal{S}(\mathcal{H}_A \otimes \mathcal{H}_B) $ with $\textup{dim}(\mathcal{H}_A) = \textup{dim}(\mathcal{H}_B) = d$ such that
\begin{align}
\mathfrak{D}(|\psi _{AB}\rangle) \geq \mathfrak{D}(|\psi'_{AB}\rangle)  \quad \text{and} \quad \sqrt{\mathfrak{F}(|\psi_{AB}\rangle)} - \sqrt{\mathfrak{F}(|\psi' _{AB}\rangle)} > 1-\delta
\end{align}
\end{lemma} 

\begin{proof}
Let us consider two pure states $|\psi_{AB}\rangle = \sum_i \sqrt{\lambda_i}|i\rangle_A |i\rangle_B$ and $|\psi'_{AB}\rangle = \sum_i \sqrt{\lambda'_i}|i\rangle_A |i\rangle_B$ with Schmidt coefficients 
\begin{align}
\label{Schmidt coefficeints}
\vec{\lambda} = \bigg[\underbrace{ \frac{1-\mu}{d-1},  \dots \dots, \frac{1-\mu}{d-1}}_{d-1\, \text{times}} ,\mu \bigg] \qquad \vec{\lambda}' = \bigg[ 0,  \dots, 0, \underbrace{ \left( \frac{1}{d-1} \right)^{1-\kappa} \hskip-15pt ,\dots, \left( \frac{1}{d-1} \right)^{1-\kappa} }_{(d-1)^{1-\kappa}\, \text{times}}  \bigg] \, . 
\end{align}
Using that $\mathfrak{F}(|\psi _{AB}\rangle)  = e^{-H_\infty(\psi_A)}$, the conditions $\mathfrak{F}(|\psi _{AB}\rangle) > \mathfrak{F}(|\psi' _{AB}\rangle)$ and  $\mathfrak{D}(|\psi _{AB}\rangle) \geq \mathfrak{D}(|\psi'_{AB}\rangle) $ turn into 
 \begin{align}
 h_{\text{bin}}(\mu)+(1-\mu)\log{(d-1)}= \quad H(\psi_A)& \geq H(\psi'_A)\quad  = (1-\kappa)\log{(d-1)}\\
-\log{\mu} = H_{\infty}(\psi_A) & < H_{\infty}(\psi'_A) = (1-\kappa)\log{(d-1)}
 \end{align}
for $\mu \geq1/d$. Since $ h_{\text{bin}}(\mu) \geq 0$ the first condition is satisfied whenever $\kappa \geq \mu$. Let us then fix $ \mu = \kappa$.
The second condition gives $(d-1)^{1-\kappa}>1/\kappa$ which is always satisfied for $d$ big enough. Then, the fidelities behave as 
 \begin{align}
\mathfrak{F}(|\psi_{AB}\rangle) \sim  \kappa  \quad \qquad \mathfrak{F}(|\psi'_{AB}\rangle) \sim \left( \frac{1}{d-1}\right)^{1-\kappa} \, . 
 \end{align}
 Then we can choose $\kappa \rightarrow 1$ and $d$ big enough such that the fidelity gap $\sqrt{\mathfrak{F}(|\psi_{AB}\rangle)} - \sqrt{\mathfrak{F}(|\psi' _{AB}\rangle)}$ is arbitrarily close to $1$. As an explicit, example for $d=3$ we can choose $\vec{\lambda}=[2/3,1/6,1/6]$ and $\vec{\lambda}'=[0,1/2,1/2]$.
\end{proof}
In this setting we identify the relative entropy of entanglement as the relevant quantity to quantify the entanglement needed for the catalyst. From Theorem~\ref{fidelity theorem 1}, we obtain
\begin{align}
\mathfrak{D}(|\nu \rangle) = \Omega \left( \log{\frac{1}{\varepsilon}}\right) \, . 
\end{align}
Therefore, to perform correlated catalytic transformation we would need a catalyst with a diverging amount of entanglement as the error approaches zero. 

\section{Resource theory of coherence}
\label{Resource theory of coherence}
In this section we  first introduce resource theory of coherence and then derive the consequences of Theorem~\ref{fidelity theorem 1} in this framework. Coherence is defined with respect a particular basis dictated by the physical problem under consideration~\cite{Winter}. If $\{|i\rangle, i=1, ...,d\} $ is such a basis, a state is called free if it is diagonal in this basis, namely if it is of the form $ \sum p_i |i\rangle \langle i |$ with $\sum p_i=1$. We  call these states \textit{incoherent states} and we denote this set (free set) as usual by $\mathcal{F}$. States that are not \textit{incoherent states} are resourceful and we refer to them as \textit{coherent states}.  We introduce the dephasing operator $\Delta(\cdot) = \sum_i |i\rangle \langle i | \cdot |i\rangle \langle i |$. In the following we refer to the quantity $\mathfrak{F}(\rho) = \max_{\sigma \in \mathcal{F}} F(\rho, \sigma)$ as fidelity of coherence. To apply the results of Theorem~\ref{fidelity theorem 1} we first need to prove that the fidelity of coherence $\mathfrak{F}$ is multiplicative. This property has already been proved in~[\citenum{Hayashi}, Theorem 3]. Here we provide an alternative proof by giving an Aberti's form for this quantity through semi-definite program (SDP) formulation. 

We first find that
\begin{theorem}
\label{SDP}
Let $\rho \in \mathcal{S}(\mathcal{H})$, then the fidelity of coherence is the solution of the following minimisation problem 
\begin{align}
\max \limits_{\sigma \in \mathcal{F}} F(\rho, \sigma) = \inf \limits_{R > 0} \textnormal{Tr}\, [\rho R^{-1}] \|\Delta(R)\|_{\infty}
\end{align}
\end{theorem}

\begin{proof}
First, we note that $\max_{\sigma \in \mathcal{F}} F(\rho, \sigma) = \max_{\sigma \in \mathcal{S}(\mathcal{H})} F(\rho, \Delta(\sigma))$.
Using the well-known SDP formulation of the square root fidelity~\cite{Watrous2}, we can write the square root fidelity of coherence as the solution of the following SDP problem  
\begin{align}
\begin{aligned}
\text{maximize}: \quad & \frac{1}{2}\text{Tr}[Z+Z^\dagger]\\
\text{subject to}: \quad & \begin{pmatrix}
\rho & Z \\
Z^\dagger & \Delta(\sigma) 
\end{pmatrix} \geq 0\\
  & Z \in \mathcal{L}(\mathcal{H}), \quad  \sigma \geq 0, \text{Tr}(\sigma) = 1 \,.   \\
\end{aligned}
\end{align}
We want to bring this in standard form, hence a maximization over $X \geq 0 $ of the functional $\text{Tr}[XA]$ subject to the constraint $\Phi(X) = B$. We set 
\begin{align}
X = \begin{pmatrix}
X_{11} & Z & \cdot \\
Z^\dagger & X_{22} & \cdot \\
\cdot & \cdot & \sigma
\end{pmatrix}\,,
\qquad 
A = \frac{1}{2}
\begin{pmatrix}
0 & \mathds{1} & 0 \\
 \mathds{1} & 0 & 0 \\
0 & 0 & 0
\end{pmatrix}\,,
\qquad 
B = 
\begin{pmatrix}
\rho & 0 & 0 \\
0 & 0 & 0 \\
0 & 0 & 1
\end{pmatrix}
\end{align}
as well as
\begin{align}
\Phi(X) = \begin{pmatrix}
X_{11} & 0 & 0 \\
0 & X_{22} - \Delta(\sigma) & 0 \\
0 & 0 & \text{Tr}[\sigma]
\end{pmatrix}\,.
\end{align}
The dual SDP is a minimization over self-adjoint $Y$ of the functional $\text{Tr}[YB]$ subject to $\Phi^\dagger(Y) \geq A$. The dual variables and the adjoint map are
\begin{align}
Y = \begin{pmatrix}
L & \cdot & \cdot \\
\cdot & R & \cdot \\
\cdot & \cdot & Q
\end{pmatrix}
\qquad \text{with} \qquad 
\Phi^\dagger(Y) = \begin{pmatrix}
L & 0 & 0 \\
0 & R  & 0 \\
0 & 0 & -\Delta(R) + \mathds{1}Q
\end{pmatrix}\, , 
\end{align}
since the dephasing channel is self-adjoint, namely $\Delta^\dagger = \Delta$. 
This leads to the following minimization problem
\begin{align}
\begin{aligned}
\text{minimize}: \quad & \text{Tr}[\rho L] + Q\\
\text{subject to}: \quad  & L,R \in \mathcal{H}(\mathcal{X}), \,  Q \in \mathbb{R} \\
  &  \mathds{1}Q \geq \Delta(R)  \\
  & \begin{pmatrix}
L & 0 \\
0 & R
\end{pmatrix} \geq 
\frac{1}{2}
\begin{pmatrix}
0 & \mathds{1} \\
\mathds{1} & 0
\end{pmatrix}\,.\\
\end{aligned}
\end{align}
The Slater condition for strong duality is satisfied. Indeed, the operator 
\begin{align}
\begin{pmatrix}
\mathds{1} & 0 & 0 \\
0 & \mathds{1} & 0 \\
0 & 0 & a
\end{pmatrix}
\end{align}
with $a>1$ is strictly feasible for the dual problem since it satisfies $\Phi^\dagger(Y) > A$.
By rescaling $L \rightarrow \frac{1}{2} L , R \rightarrow \frac{1}{2} R, Q \rightarrow \frac{1}{2} Q$ and using that~\cite{Watrous2}
\begin{align}
\begin{pmatrix}
L & -\mathds{1} \\
-\mathds{1} & R
\end{pmatrix} \geq 0 \Longleftrightarrow 
L,R > 0 , L \geq R^{-1}\,,
\end{align}
we can choose $L = R^{-1}$ without loss of generality and our problem simplifies
\begin{align}
\begin{aligned}
\text{minimize}: \quad & \frac{1}{2} \text{Tr}[\rho R^{-1}] + \frac{1}{2} Q\\
 \text{subject to}: \quad & R > 0 , \,  Q > 0 \\
  &  \mathds{1}Q \geq \Delta(R)  \, .\\
\end{aligned}
\end{align}
Following the argument leading to Alberti's expression for the fidelity~\cite{Watrous2}, going back from root fidelity to fidelity again and using that since we minimise over $\mathds{1}Q \geq \Delta(R)$, by definition of the infinity norm, we can set  $Q  = \|\Delta(R)\|_\infty$ without loss of generality, we obtain
\begin{align}
\max \limits_{\sigma \in \mathcal{F}} F(\rho, \sigma) = \inf \limits_{R > 0} \text{Tr}[\rho R^{-1}] \|\Delta(R)\|_{\infty}\, . 
\end{align}
\end{proof}
From the previous lemma we recover multiplicativity of the fidelity of coherence, as first established in \cite{Hayashi}. 
\begin{lemma}[Multiplicativity of the fidelity of coherence]
\label{multiplicativity}
For any $\rho \in \mathcal{S}(\mathcal{H}_A)$ and $\tau \in \mathcal{S}(\mathcal{H}_B)$ we have
\begin{align}
\mathfrak{F}(\rho \otimes \tau) = \mathfrak{F}(\rho) \cdot \mathfrak{F}(\tau)
\end{align}
\end{lemma}

\begin{proof}
We first prove $\mathfrak{F}(\rho \otimes \tau) \geq \mathfrak{F}(\rho) \cdot \mathfrak{F}(\tau)$. We call $\mathcal{F}_A$ and $\mathcal{F}_B$ the set of free states of the Hilbert spaces $\mathcal{H}_A$ and $\mathcal{H}_B$, respectively. 
The inequality follows immediately by noting that if $\sigma_A \in \mathcal{F}_A$ and $\sigma_B \in \mathcal{F}_B$ then $\sigma_A \otimes \sigma_B \in \mathcal{F}_{AB}$ and that the Uhlmann fidelity is multiplicative under tensor products. 

The opposite inequality $\mathfrak{F}(\rho \otimes \sigma) \leq \mathfrak{F}(\rho) \cdot \mathfrak{F}(\sigma)$ can be proved using semidefinite programming duality. By Theorem~\ref{SDP} we have $\mathfrak{F}(\rho)= \inf_{R > 0} \text{Tr}[\rho R^{-1}] \|\Delta(R)\|_{\infty}$. Let us call $\tilde{R}_A$ and $\bar{R}_B$ the optimizers such that $\mathfrak{F}(\rho_A) = \text{Tr}[\rho_A \tilde{R}_A^{-1}]\|\Delta(\tilde{R}_A)\|_{\infty}$ and  $\mathfrak{F}(\tau_B) = \text{Tr}[\tau_B \bar{R}_B^{-1}]\|\Delta(\bar{R}_B)\|_{\infty}$.
Then $\tilde{R}_A \otimes \bar{R}_B$ is a feasible operator  for $\mathfrak{F}(\rho_A \otimes \tau_B)$ since if $\tilde{R}_A , \bar{R}_B > 0$ then $\tilde{R}_A \otimes \bar{R}_B>0$.
Using multiplicativity of the infinite norm under tensor products we have $\| \Delta(\tilde{R}_A \otimes \bar{R}_B) \|_\infty = \| \Delta(\tilde{R}_A) \otimes \Delta(\bar{R}_B)) \|_\infty  =  \| \Delta(\tilde{R}_A) \|_\infty \cdot \|\Delta(\bar{R}_B)) \|_\infty $ and therefore
\begin{align}
\mathfrak{F}(\rho_A \otimes \tau_B) & \leq \text{Tr}[\rho_A \otimes \tau_B \tilde{R}_A^{-1} \otimes \bar{R}_B^{-1}]\|\Delta(\tilde{R}_A)\|_{\infty} \cdot \|\Delta(\bar{R}_B)\|_{\infty} \\
& = \mathfrak{F}(\rho_A)  \cdot \mathfrak{F}(\tau_B)
\end{align}
by definition of the optimizers $\tilde{R}_A$ and $\bar{R}_B$. 
\end{proof}

In this framework, the relative entropy of coherence $\mathfrak{D}(\rho) = \min_{\sigma \in \mathcal{F}} D(\rho \| \sigma)$ is a resource monotone that is tensor product additive, continuous and superadditive~\cite{Winter, Superadditivity}.
Moreover, if $\rho$ can be asymptotically mapped into $\rho'$ , then, under some mild assumptions which are easily seen to be satisfied in this context, $\rho$ can be also transformed into $\rho'$ under correlated catalytic transformation~\cite{Takagi}. Since a mixed state $\rho$ can be asymptotically transformed into a pure state $|\phi\rangle$ with unit rate if $\mathfrak{D}(\rho) \geq \mathfrak{D}(|\phi\rangle)$~\cite{Winter}, we conclude that $\rho$ is transformable into $|\phi\rangle$ by a correlated catalytic transformation if and only if the relative entropies of coherence are ordered, namely if $\mathfrak{D}(\rho) \geq \mathfrak{D}(|\phi\rangle)$. Therefore we look for states that satisfy both $\mathfrak{F}(\rho) > \mathfrak{F}(|\phi\rangle)$ and $\mathfrak{D}(\rho) \geq \mathfrak{D}(|\phi\rangle)$. These states exists and in addition, we can always find an Hilbert space big enough such that the fidelity gap $\sqrt{\mathfrak{F}(\rho)}-\sqrt{\mathfrak{F}(|\phi\rangle)}$ is arbitrarily close to one. Indeed we find 
\begin{lemma}
For any $\delta > 0$, there exist $d > 0$ and two states $\rho, |\phi\rangle \in \mathcal{S}(\mathcal{H})$ such that 
\begin{align}
\mathfrak{D}(\rho) \geq \mathfrak{D}(|\phi\rangle)  \quad \text{and} \quad \sqrt{\mathfrak{F}(\rho)}  - \sqrt{\mathfrak{F}(|\phi\rangle)} > 1-\delta
\end{align}
\end{lemma}

\begin{proof}
Let us consider the following states 
\begin{align}
\rho = \mu \oplus (1-\mu) |\Phi_{d-1} \rangle \langle \Phi_{d-1}| \qquad \qquad |\phi\rangle = |d_1\rangle \otimes |\Phi_{d_2} \rangle  \qquad \text{with} \quad d_1 \sim d^{1-\varepsilon}\, , \,d_2 \sim d^\varepsilon \, , 
\end{align}
where $| \Phi_d \rangle = \frac{1}{\sqrt{d}} \sum_{i=0}^{d-1} |i \rangle $ is the maximally coherent pure state and $ \varepsilon > 0$ is some fixed small constant. 
Then, using that $\mathfrak{D}(\rho) = S(\Delta(\rho)) - S(\rho) $ it can be easily found that the relative entropies of coherence for $d \gg 1$ scale as
\begin{align}
(1-\mu) \log{d} \quad \sim \quad \mathfrak{D}(\rho) \geq \mathfrak{D}(|\phi\rangle) \quad \sim \quad  \varepsilon \log{d} \, . 
\end{align}
The last inequality is satisfied for $\mu \leq 1- \varepsilon$. 

The fidelity of coherence of the direct of block diagonal state $a  \oplus b$ is equal to
$\mathfrak{F}(a \oplus b) = \max_{\Delta_1, \Delta_2, \\ \text{Tr}(\Delta_1+\Delta_2)=1 } (\sqrt{F(a, \Delta_1)} + \sqrt{F(b, \Delta_2)})^2$ where $\Delta_1, \Delta_2$ are diagonal matrices . By recalling that $\mathfrak{F}(|\Phi_d \rangle) =1/d$~\cite{Hayashi} we can easily find that
\begin{align}
 \mathfrak{F}(\rho) \geq \left( \mu + \frac{1-\mu}{\sqrt{d-1}} \right)^2 \quad > \quad \mathfrak{F}(|\phi\rangle)\quad \sim \quad \frac{1}{d^\varepsilon}\, , 
\end{align} 
so that for $d$ big enough we can choose $\mu \sim 1$ and the fidelity gap $\sqrt{\mathfrak{F}(\rho)}-\sqrt{\mathfrak{F}(\rho')} \sim 1$.

As an explicit example, we mention that for $d = 4$ we can choose $\rho = \mu \oplus (1-\mu) |\Phi_3 \rangle \langle |\Phi_3 \rangle$ and $|\phi\rangle = |1 \rangle \otimes |+ \rangle$ and set $\mu = 1 - 1/\log{3}$. 
These two states satisfy $\mathfrak{D}(\rho) = \mathfrak{D} (|\phi\rangle)$ and $\mathfrak{F}(\rho) > \mathfrak{F} (|\phi\rangle)$. 
\end{proof}
Note that for any state $\rho$ it holds $\mathfrak{F}(\rho)  \geq 1/d$. To see that, it is sufficient to choose the free state $\sigma = \mathds{1}/d \in \mathcal{F}$ and notice that $\text{Tr}(\sqrt{\rho}) \geq 1$ for any $\rho \in \mathcal{S}(\mathcal{H})$. 
Therefore, using Theorem~\ref{fidelity theorem 1} we can immediately bound the dimensions the catalyst
\begin{align}
d = \Omega \left(\frac{1}{\varepsilon}\right) \, . 
\end{align}
Therefore the dimension of the catalyst, as we found for resource theory of athermality,  must diverge as the error vanishes.
We identify the relative entropy of coherence $\mathfrak{D}(\rho) = \min_{\sigma \in \mathcal{F}} D(\rho \| \sigma)$ as the relevant physical quantity in this framework. From Theorem~\ref{fidelity theorem 1} we find
\begin{align}
\mathfrak{D}(\nu) = \Omega\left( \log{\frac{1}{\varepsilon}}\right) \, . 
\end{align}
Hence, similarly to the previous case, a vanishing error in the transformation implies the catalyst to have a diverging amount of coherence.

\appendix

\section{A tighter bound for the $\alpha=1/2$ case}
\label{tighter}
We prove a tighter bound for the main theorem in the case $\alpha = 1/2$. 

\begin{theorem}
Assume that $\rho, \rho' \in \mathcal{S}(\mathcal{H})$, $\mathfrak{D}_{1/2}$ is additive for the state $\rho'$ and $\mathfrak{D}_{1/2}(\rho) < \mathfrak{D}_{1/2}(\rho')$. Then, for any $\varepsilon$-correlated catalytic transformation with catalyst $\nu$ mapping $\rho$ into $\rho'$, we have
\begin{align}
\sqrt{\mathfrak{F}(\nu)} \leq \frac{\varepsilon}{\sqrt{\mathfrak{F}(\rho)}-\sqrt{\mathfrak{F}(\rho')}}\,.
\end{align}
\end{theorem}

\begin{proof}
The first part of the proof is the same as the main proof given in Section~\ref{Proof of the main Theorem} for $\alpha = 1/2$. 
It thus remains to find an upper bound on $f$. We get
\begin{align}
\Delta \mathfrak{D}_\alpha \leq f \leq \widetilde{D}_{\frac{1}{2}}(\eta \| \sigma^*_\rho \otimes \sigma^*_\nu) - \mathfrak{D}_{\frac{1}{2}}(\rho \otimes \nu) = -\log{\frac{F(\eta \| \sigma^*_\rho \otimes \sigma^*_\nu)}{\mathfrak{F}(\rho \otimes \nu)}} \, . 
\end{align}
We then use the tighter triangular inequality for the purified distance $P(\rho,\tau) \leq P(\rho,\sigma) \sqrt{F(\sigma, \tau)}$ $+ P(\sigma,\tau)\sqrt{F(\rho, \sigma)}$ which holds for $P(\rho,\sigma)^2 + P(\sigma, \tau)^2 \leq 1$~[\citenum{Tomamichel}, Proposition 3.16].  
We also introduce the parameter $\sqrt{\gamma} := P(\eta,\rho \otimes \nu) \leq \varepsilon$.
If $\mathfrak{F}(\rho \otimes \nu) \leq \gamma$ the statement follows trivially. If $\mathfrak{F}(\rho \otimes \nu) \geq \gamma$, the condition $P(\eta,\rho \otimes \nu)^2 + P(\rho \otimes \nu, \sigma^*_\rho \otimes \sigma^*_\nu)^2 \leq 1$ holds and hence we can apply the tighter triangular inequality. We get 
\begin{align}
 \Delta \mathfrak{D}_{\frac{1}{2}} \leq  f & \leq -\log{\frac{1-(\sqrt{\gamma \mathfrak{F}(\rho \otimes \nu)}+\sqrt{(1-\gamma)(1-\mathfrak{F}(\rho \otimes \nu))})^2}{\mathfrak{F}(\rho \otimes \nu)}}\\
&= -\log{\left(1- \gamma + \gamma x^2 -2 \sqrt{\gamma(1-\gamma)}x \right)} \, , 
\end{align}
where we set $x = \sqrt{\frac{1-\mathfrak{F}(\rho \otimes \nu)}{\mathfrak{F}(\rho \otimes \nu)}}$. By solving the equation in $x$ and rewriting the solution in terms of the fidelity we get 
\begin{align}
\mathfrak{F}(\rho \otimes \nu) \leq \frac{\gamma}{1+2^{- \Delta \mathfrak{D}_{\frac{1}{2}}}-2\sqrt{(1-\gamma)}2^{-\frac{ \Delta \mathfrak{D}_{\frac{1}{2}}}{2}}}\,.
\end{align}
Using that $2^{- \Delta \mathfrak{D}_{\frac{1}{2}}} = \frac{\mathfrak{F}(\rho')}{\mathfrak{F}(\rho)}$, the inequality $\mathfrak{F}(\rho \otimes \nu) \geq \mathfrak{F}(\rho)\cdot \mathfrak{F}(\nu)$ and the relation $\sqrt{\gamma} \leq \varepsilon$, we get the following upper bound for the fidelity of the catalyst
\begin{align}
\sqrt{\mathfrak{F}(\nu)} \leq \frac{\varepsilon}{\sqrt{\mathfrak{F}(\rho)}-\sqrt{\mathfrak{F}(\rho')}}\,.
\end{align}
\end{proof}

\section{Remarks on smoothing}
\label{remark on smoothing}
In this appendix we argue why the optimisation over sub-normalised states is necessary for the smoothed sandwiched R\'enyi divergences with $\alpha \in [1/2,1)$ to be invariant under embedding in a larger space. Moreover, we show that for $\alpha \in [0,1)$ it is not possible to define smoothed Petz R\'enyi divergences that satisfy data-processing. 
We consider the states 
\begin{equation}
\rho^{(2)} =
\begin{pmatrix}
1 & 0 \\
0 & 0
\end{pmatrix} \qquad 
\sigma^{(2)} =
\begin{pmatrix}
\frac{1}{2} & 0 \\
0 & \frac{1}{2}
\end{pmatrix}
\end{equation}
as well as their embeddeding in a three dimensional space 
\begin{equation}
\rho^{(3)} = \begin{pmatrix}
1 & 0 & 0\\
0 & 0 & 0\\
0 & 0 & 0
\end{pmatrix} \qquad 
\sigma^{(3)} = \begin{pmatrix}
\frac{1}{2} & 0 & 0\\
0 & \frac{1}{2} & 0\\
0 & 0 & 0
\end{pmatrix}
\end{equation}
and the pure state  $|\phi \rangle = (\sqrt{1-\varepsilon^2}, 0 ,\sqrt{\varepsilon^2})$. In the following we use the shorthands $\circ$ and $\bullet$ to indicate that the smoothing is defined over the normalised and sub-normalised states, respectively. 
We then find for $\alpha \in [1/2,1)$
\begin{align}
\label{1}
\circ: \quad  &\widetilde{D}_{\alpha}^\varepsilon(\rho^{(2)} \| \sigma^{(2)}) =  \log{2}\\
\label{2}
\circ: \quad &\widetilde{D}_{\alpha}^\varepsilon(\rho^{(3)} \| \sigma^{(3)}) =\log{2} - \frac{\alpha}{1-\alpha} \log{(1-\varepsilon^2)} \\
\label{3}
\bullet: \quad &\widetilde{D}_{\alpha}^\varepsilon(\rho^{(2)} \| \sigma^{(2)}) =  \log{2} - \frac{\alpha}{1-\alpha} \log{(1-\varepsilon^2)} \\
\label{4}
\bullet: \quad  &\widetilde{D}_{\alpha}^\varepsilon(\rho^{(3)} \| \sigma^{(3)}) =  \log{2} - \frac{\alpha}{1-\alpha} \log{(1-\varepsilon^2)}
\end{align}

For any state $\tilde{\rho}$ in the $\varepsilon$-ball of $\rho$ we have $F(\rho, \tilde{\rho}) \geq 1 - \varepsilon^2$. To prove \eqref{1} we note that since $\rho^{(2)}$ is pure we can obtain any mixed state by just a rotating $\rho^{(2)}$ in the Bloch sphere and then applying a depolarizing channel. Given that the fully mixed state $\sigma^{(2)}$ is invariant under such quantum channel, by data-processing it follows that, if we were to consider smoothing only on normalized states, $\rho^{(2)}$ would achieve the maximum value $\widetilde{D}_{\alpha}^\varepsilon(\rho^{(2)} \| \sigma^{(2)}) = \widetilde{D}_{\alpha}(\rho^{(2)} \| \sigma^{(2)}) = \log{2}$. 

The pure state $|\phi \rangle = (\sqrt{1-\varepsilon^2}, 0 ,\sqrt{\varepsilon^2})$ satisfies $F(|\phi \rangle,\rho^{(3)}) = 1-\varepsilon^2$ and hence it is in the $\varepsilon$-ball of $\rho^{(3)}$. Moreover,  $\widetilde{D}_{\alpha}^\varepsilon(|\phi\rangle \| \sigma^{(3)}) = \log{2} - (\alpha/(1-\alpha))\log{(1-\varepsilon^2)} > \log{2}$ and therefore smoothing in three dimensions on normalised states achieves a bigger value than in two dimensions. However, if we assume subnormalized states then $(1-\varepsilon^2)|0\rangle \langle 0|$ (which would also be the optimser) achieves the same value, namely $\widetilde{D}_{\alpha}^\varepsilon(\rho^{(2)} \| \sigma^{(2)}) = \widetilde{D}_{\alpha}((1-\varepsilon^2)|0\rangle \langle 0| \| \sigma^{(2)}) = \log{2} - (\alpha/(1-\alpha)) \log{(1-\varepsilon^2)}$. To prove that it is the optimiser we first notice that we need to minimise the functional $\min \text{Tr}(\tilde{\rho}^{(2)})^\alpha$ subject to the condition of the $\varepsilon$-ball $\langle 0|\tilde{\rho}^{(2)} |0\rangle \geq 1-\varepsilon^2$. Using that $\tilde{\rho}^{(2)} \leq \max_i \lambda_i \mathds{1}$ where $\lambda_i$ are its eigenvalues, the condition of the $\varepsilon$-ball gives $\max_i \lambda_i \geq 1-\varepsilon^2$ which implies for the functional  $ \text{Tr}(\tilde{\rho}^{(2)})^\alpha = \sum_i \lambda^\alpha_i \geq  (1-\varepsilon^2)^\alpha$ which is achieved by $(1-\varepsilon^2)|0\rangle \langle 0|$ (equation \eqref{2}).

By data-processing, we have that for subnormalized states (see Section \ref{section data-processing}) the sandwiched R\'enyi divergences are invariant under embedding in a larger space. Hence $\widetilde{D}_{\alpha}^\varepsilon(\rho^{(3)} \| \sigma^{(3)}) = \widetilde{D}_{\alpha}^\varepsilon(\rho^{(2)} \| \sigma^{(2)}) = \log{2} - (\alpha/(1-\alpha)) \log{(1-\varepsilon^2)}$ and in particular $|\phi \rangle$ achieves the maximum value also if we were to consider subnormalized states in the optimisation (equations \eqref{2} and \eqref{3}). 
In conclusion, to obtain a well-defined quantity, which is invariant under embedding in a larger space, we need to consider subnormalized states.

We now consider the Petz R\'enyi divergences. Let $\alpha \in (0,1) \cup (1,\infty)$ and $\rho$ and $\sigma$ positive operators with $\rho \neq 0$. The \emph{Petz R\'enyi divergence} of $\sigma$ with $\rho$ is~\cite{petz1986quasi,Tomamichel}
\begin{equation}
\bar{D}_{\alpha}(\rho \| \sigma) :=
\begin{cases}
\frac{1}{\alpha-1} \log{\text{Tr}(\rho^\alpha \sigma^{1-\alpha})} & \text{if}\;(\alpha<1 \wedge \rho \not \perp \sigma) \vee \rho \ll \sigma \\
 +\infty & \text{else}
\end{cases}
\end{equation}
Moreover $\bar{D}_0$ and $\bar{D}_1$ are defined as the respective limits of $\bar{D}_\alpha$ for $\alpha \rightarrow \{0,1\}$. The Petz R\'enyi divergence satisfies the data-processing inequality for $\alpha \in [0,2]$. 

Hence, analogously to the sandwiched case, we could define the smoothed Petz R\'enyi divergences
\begin{equation}
    \bar{D}^\varepsilon_{\alpha}(\rho \| \sigma)=
    \begin{cases}
      \max \limits_{\tilde{\rho} \in B^{\varepsilon}(\rho)} \bar{D}_{\alpha}(\tilde{\rho}\|\sigma), & \text{if}\ \alpha \in [0,1) \\
       \min \limits_{\tilde{\rho} \in B^{\varepsilon}(\rho)} \bar{D}_{\alpha}(\tilde{\rho}\|\sigma), & \text{if}\ \alpha \in (1,2]
    \end{cases}
  \end{equation}
For $\alpha \in (1,2]$ the data-processing inequality follows trivially. 
We then find for $\alpha \in [0,1)$
\begin{align}
\label{1.1}
\circ: \quad  &\bar{D}_{\alpha}^\varepsilon(\rho^{(2)} \| \sigma^{(2)}) =  \log{2}\\
\label{2.1}
\circ: \quad &\bar{D}_{\alpha}^\varepsilon(\rho^{(3)} \| \sigma^{(3)}) \geq \log{2} - \frac{1}{1-\alpha} \log{(1-\varepsilon^2)} \\
\label{3.1}
\bullet: \quad &\bar{D}_{\alpha}^\varepsilon(\rho^{(2)} \| \sigma^{(2)}) =  \log{2} - \frac{\alpha}{1-\alpha} \log{(1-\varepsilon^2)}
\end{align}
We have $\bar{D}_{\alpha}^\varepsilon(\rho^{(2)} \| \sigma^{(2)}) = \log{2}$ if we were to consider the smoothing only on normalized states and  $\bar{D}_{\alpha}(|\phi \rangle \| \sigma^{(3)}) = \log{2} - (1/(1-\alpha)) \log{(1-\varepsilon^2)}$ (equations \eqref{1.1} and \eqref{2.1}). However, if we were to allow subnormalized states, then $(1-\varepsilon^2)|0\rangle \langle 0|$ for what we have already discussed would achieve the maximum in  $\bar{D}_{\alpha}^\varepsilon(\rho^{(2)} \| \sigma^{(2)}) = \bar{D}_{\alpha}((1-\varepsilon^2)|0\rangle \langle 0| \| \sigma^{(2)}) = \log{2} - (\alpha/(1-\alpha)) \log{(1-\varepsilon^2)} < \log{2} - (1/(1-\alpha)) \log{(1-\varepsilon^2)}$ (equation \eqref{3.1}) and therefore also by allowing subnormalized states the smoothed Petz R\'enyi divergences for $\alpha \in [0,1)$ would not be invariant under embedding in a larger space and in particular they wouldn not satisfy the data-processing inequality under isometries.

\section{First order expansion for error exponent in asymptotic pairwise state transformation for small relative entropy gaps}
\label{error exponent}
In this appendix we qualitatively derive a first order lower bound for the error exponent $\gamma$ in approximate asymptotic pairwise state transformation \cite{Buscemi}. We give an expansion in terms of the relative entropy gap $\Delta D := D(\rho_1\|\sigma_1)-D(\rho_2\|\sigma_2)$; the higher order corrections are arbitrarily small for a sufficiently small gap. In \cite{Buscemi} the authors proved that if $ D(\rho_1 \|\sigma_1)> D(\rho_2\|\sigma_2) $, then there exists a sequence of channels that transforms $(\rho_1^{\otimes n}, \sigma_1^{\otimes n})$ to $(\rho_2^{\otimes n}, \sigma_2^{\otimes n})$ where the transformation  $\rho_1^{\otimes n} \rightarrow \rho_2^{\otimes n}$ has an exponentially vanishing error $\varepsilon_n \leq 2^{-\gamma n} $  with the number of copies $n$ while the second transformation $\sigma_1^{\otimes n} \rightarrow \sigma_2^{\otimes n}$ is exact. We call $\gamma$ error exponent. In the following we assume that all states have full support and that $ V(\rho_1\|\sigma_1), V(\rho_2\|\sigma_2)>0$ where the variance is defined as $V(\rho\|\sigma) := \text{Tr}[\rho(\log{\rho}-\log{\sigma})^2)]-D(\rho\|\sigma)^2$. Moreover, we assume that for sufficiently small gaps the functions are well behaved.  We then find
\begin{proposition}
Let  $\rho_1,\rho_2, \sigma_1, \sigma_2$ be four quantum states such that $ D(\rho_1 \|\sigma_1)>D(\rho_2\|\sigma_2)$. Then the error exponent for the asymptotic transformation in the iid case satisfies:
\begin{equation}
\label{lower bound gamma}
\gamma \geq \frac{\Delta D^2 \log{e}}{8(V_1+V_2)}+O(\Delta D^3) 
\end{equation}
where $\Delta D := D(\rho_1 \|\sigma_1) - D(\rho_2\|\sigma_2)$ and $V_i := V(\rho_i\|\sigma_i)$. 
\end{proposition}

\begin{proof}
The following equalities hold~[\citenum{audenaert2012quantum},Proposition 3.2],[\citenum{anshu2019minimax}, Theorem 3]:
\begin{align}
&D_h^\varepsilon(\rho^{\otimes n} \| \sigma^{\otimes n} ) \geq n\widetilde{D}_{\alpha}(\rho \| \sigma) - \frac{\alpha}{1-\alpha} \log{\frac{1}{\varepsilon}} \qquad \qquad \qquad \qquad  \quad  \, \, \, \hskip5pt \alpha \in [0,1) \\ 
&D_{\max}^{\varepsilon, \Delta} (\rho \| \sigma) \leq \widetilde{D}_{\alpha}(\rho \| \sigma) + \frac{1}{\alpha-1} \log{\frac{1}{\varepsilon^2}} + \log{\frac{1}{1-\varepsilon^2}}  \quad \quad \,\,\,\,\, \qquad \quad  \alpha \in (1, \infty] 
\end{align}
For any $n$, a map that transforms  $(\rho_1^{\otimes n}, \sigma_1^{\otimes n})$ to $(\rho_2^{\otimes n}, \sigma_2^{\otimes n})$  where only the first transformation is approximated exists if $ D_{h}^{\varepsilon_n}(\rho_1^{\otimes n}\| \sigma_1^{\otimes n})-D_{\max}^{\varepsilon_n}(\rho_2^{\otimes n}\| \sigma_2^{\otimes n}) \geq 0$ \cite{Buscemi}. 
Using the above expansions we get for some small $\delta_1, \delta_2>0$ :
 \begin{align}
 D_{h}^{\varepsilon_n}(\rho^{\otimes n}_1\| \sigma^{\otimes n}_1)-D_{\max}^{\varepsilon_n}(\rho^{\otimes n}_2\| \sigma^{\otimes n}_2))  \geq & n\widetilde{D}_{1-\delta_1}(\rho_1\| \sigma_1)+\frac{1-\delta_1}{\delta_1}\log{\varepsilon_n}
 &\quad   \notag \\& - n\widetilde{D}_{1+\delta_2}(\rho_2\| \sigma_2)+\frac{2}{\delta_2}\log{\varepsilon_n}+\log{(1-\varepsilon^2_n)} 
 \end{align}
We define $\widetilde{D}_{1-\delta_1}(\rho_1\| \sigma_1) - \widetilde{D}_{1+ \delta_2}(\rho_2\| \sigma_2) := \kappa$.
We then get 
 \begin{align}
 D_{h}^{\varepsilon_n}(\rho^{\otimes n}_1\| \sigma^{\otimes n}_1)-D_{\max}^{\varepsilon_n}(\rho^{\otimes n}_2\| \sigma^{\otimes n}_2))
 & \geq n \kappa +\left( \frac{1-\delta_1}{\delta_1}+\frac{2}{\delta_2} \right)\log{\varepsilon_n}
 \end{align}
 where we used $1-\varepsilon^2_n\geq \varepsilon_n$ since we assumed $\varepsilon_n \leq 1/2$.
Therefore we can set
 \begin{equation}
 \varepsilon_n=2^{-n\kappa \left( \frac{1}{\delta_1}+\frac{2}{\delta_2} \right)^{-1}} 
 \end{equation}
We then use the expansions $\delta_i=\frac{\Delta D-\kappa}{2}\frac{2\log{e}}{V(\rho_i\|\sigma_i)} + O((\Delta D-\kappa)^2)$ and we maximise the expression over $\kappa \in (0,\Delta D)$ (the maximum is achieved at $\kappa = \Delta D /2$). We then get 
 \begin{equation}
 \varepsilon_n = 2^{-n\left[ \frac{\Delta D^2 \log{e}}{4(V_1+2V_2)}+O(\Delta D^3) \right]}
 \end{equation}
where we set $V_i := V(\rho_i\|\sigma_i)$ which concludes the proof. 
\end{proof}

We remark that the  the first order term in \eqref{lower bound gamma} has the right scaling with the number of copies; indeed if we consider the transformation $\rho_i \rightarrow \rho_i^{\otimes a} $ and $\sigma_i \rightarrow \sigma_i^{\otimes a} $ for some integer $a$ then we must have $\gamma \rightarrow a \gamma$.

\end{document}